\documentclass[journal,comsoc]{IEEEtran}
\usepackage[cmex10]{amsmath} % Use the [cmex10] option to ensure complicance
\usepackage[T1]{fontenc}% optional T1 font encoding
\usepackage[table]{xcolor}

\usepackage{amssymb,amsthm}
\usepackage{url}            % simple URL typesetting
\usepackage{booktabs}       % professional-quality tables
\usepackage{nicefrac}       % compact symbols for 1/2, etc.
\usepackage{amsmath,amssymb,amsfonts}
\DeclareMathOperator*{\argmax}{arg\,max}

\usepackage[cmintegrals]{newtxmath}
\usepackage{mathrsfs} %\mathscr
\usepackage{mathtools} %\mathclap

\usepackage{enumitem}
\usepackage{graphicx, tikz}
\newcommand\sbullet[1][.5]{\mathbin{\vcenter{\hbox{\scalebox{#1}{$\bullet$}}}}}

\usepackage{subcaption}
\usepackage{bm} %for \bm, mathbold command
\usepackage{relsize}
\newcommand{\bpi}{{\mathlarger{\mathlarger{\pi}}}}

\newcommand{\bz}{{\mathlarger{\mathlarger{z}}}}
\usepackage{breqn}
\usepackage[htt]{hyphenat}
\usepackage{algpseudocode,algorithm,algorithmicx}

\allowdisplaybreaks

\newcommand{\ddmodel}{\texttt{DD-model}}
\usepackage{xspace}

\newcommand{\Hquad}{\hspace{0.5em}}

\def\parsec{\par\noindent}

\newcommand{\E}{\mathbb{E}}
\renewcommand{\P}{\mathbb{P}}
\newcommand{\nn}{\nonumber}
\renewcommand{\epsilon}{\varepsilon}
\def\deg{{\text{deg}}}

\def\Aut{ {\mathrm{Aut}}}

\newcommand{\bone}{\boldsymbol{1}}
\newcommand{\R}{\mathcal{R}}
\newcommand{\NR}{\mathcal{NR}}
\newcommand{\N}{\mathcal{N}}
\newcommand{\C}{\mathcal{C}}
\newcommand{\cH}{\mathcal{H}}
\newcommand{\cG}{\mathcal{G}}
\newcommand{\cN}{\mathcal{N}}

\newcommand{\xhdr}[1]{\vspace{1.7mm}\noindent{{\bf #1.}}}
\newcommand{\xhdrLessSpace}[1]{\vspace{1 ex}\noindent{{\bf #1.}}}

\newtheorem{theorem}{Theorem}
\newtheorem{lemma}{Lemma}

\theoremstyle{definition}

\newtheorem{remark}{Remark}

\usepackage{graphicx, tikz}

\begin{document}

\title{Temporal Ordered Clustering in Dynamic Networks: Unsupervised and Semi-supervised Learning Algorithms}

\author{Krzysztof~Turowski\textsuperscript{\textdagger},
        Jithin~K.~Sreedharan\textsuperscript{\textdagger},
        and~Wojciech~Szpankowski,~\IEEEmembership{Fellow,~IEEE}% <-this % stops a space
\IEEEcompsocitemizethanks{
\IEEEcompsocthanksitem K.~Turowski is with the Theoretical Computer Science Department, Jagiellonian University, Krakow, Poland.\protect\\
E-mail: krzysztof.szymon.turowski@gmail.com.
\IEEEcompsocthanksitem J.~K.~Sreedharan and W.~Szpankowski are with the Dept.\ of Computer Science and the NSF Center for Science and Information, Purdue University, West Lafayette, IN 47907, U.S.A.\protect\\
E-mail: \{jithinks, szpan\}@purdue.edu.
}
\thanks{\textsuperscript{\textdagger} Both the authors contributed equally to this research.}%
\thanks{This work was supported by NSF Center for Science of Information (CSoI)
Grant CCF-0939370, and in addition by NSF Grants CCF-1524312, CCF-2006440,
and CCF-2007238, National Science Center Grant UMO-2016/21/B/ST6/03146 and Google Research Award.
}}

\maketitle

\begin{abstract}
In {\em temporal ordered clustering}, given a single snapshot of a dynamic network in which nodes arrive at distinct time instants, we aim at partitioning its nodes into $K$ ordered clusters $\C_1 \prec \cdots \prec \C_K$ such that for $i<j$, nodes in cluster $\C_i$ arrived before nodes in cluster $\C_j$, with $K$ being a data-driven parameter and not known upfront.
Such a problem is of considerable significance in many applications ranging from tracking the expansion of fake news to mapping the spread of information.
We first formulate our problem for a general dynamic graph, and propose an integer programming framework that finds the optimal clustering, represented as a strict partial order set, achieving the best precision (i.e., fraction of successfully ordered node pairs) for a fixed density (i.e., fraction of comparable node pairs).
We then develop a sequential importance procedure and design unsupervised and semi-supervised algorithms to find temporal ordered clusters that efficiently approximate the optimal solution.
To illustrate the techniques, we apply our methods to the vertex copying (duplication-divergence) model which exhibits some edge-case challenges in inferring the clusters as compared to other network models.
Finally, we validate the performance of the proposed algorithms on synthetic and real-world networks.
\end{abstract}


\section{Introduction}
The clustering of nodes is a classic problem in networks. In its typical form in static networks, it finds communities where methods like spectral clustering, modularity maximization, minimum-cut method, and hierarchical clustering are commonly used \cite{SCHAEFFER200727}.

However, in dynamic networks that grow over time with nodes or edges getting added or deleted, a criteria of clustering based on its temporal characteristics finds significant relevance in practice since it helps us to study the existence of certain network structures and their future behavior.
One approach to reason about the history of dynamic networks via clustering is guided by the problem of node labeling according to their arrival order when {\em only the structure} of the final snapshot of the network is provided. The availability of merely structure means that either we are given an unlabeled graph or the current node labels do not present any historical information.
%when only the final snapshot of the evolved graph is given
As it turns out, in many real-world networks and graph models, it is impossible to find a complete order of arrival of nodes due to a large number of symmetries inherent in the graph \cite{luczak2016asymmetry,turowski2018allerton}. Figure~\ref{fig:intro_example} shows an example.
In such cases, it is essential to classify nodes that are indistinguishable themselves in terms of arrival order into clusters $\{ \C_i\}$.
% In such cases, we can find a partial order $\sigma$ between the node pairs such that $u <_{\sigma} v$ indicates node $u$ arrived earlier than node $v$ in the graph.
% When two nodes are not comparable under $\sigma$, we deduce them as indistinguishable and such a partial order naturally classifies nodes into clusters.
% Such a partial order naturally translates into clusters of nodes $\{ \C_i\}$ and introduces an order among the clusters as $\C_1 \prec \C_2 \prec \cdots $ so that for any $i<j$, all the nodes in the cluster $\C_i$ are estimated to be arrived earlier than all the nodes in the cluster $\C_j$, and all the nodes inside each cluster are considered to be identical in arrival order.  We call such a clustering scheme as {\em temporal ordered clustering}.
Furthermore, the formed clusters also will be ordered as $\C_1 \prec \C_2 \prec \cdots $ so that for any $i<j$, all the nodes in the cluster $\C_i$ are estimated to be arrived earlier than all the nodes in the cluster $\C_j$, and all the nodes inside each cluster are considered to be identical in arrival order. We call such a clustering scheme as {\em temporal ordered clustering}.
%In the ideal case, where we could infer the arrival order exactly, each cluster will contain only singletons and number of clusters will be same as the number of nodes. But such a situation may not arise due to large number of symmetries typically present in real-world networks.

\begin{figure}[!htb]
    %    \vspace*{-2 em}
    \centering
\includegraphics[scale=0.25]{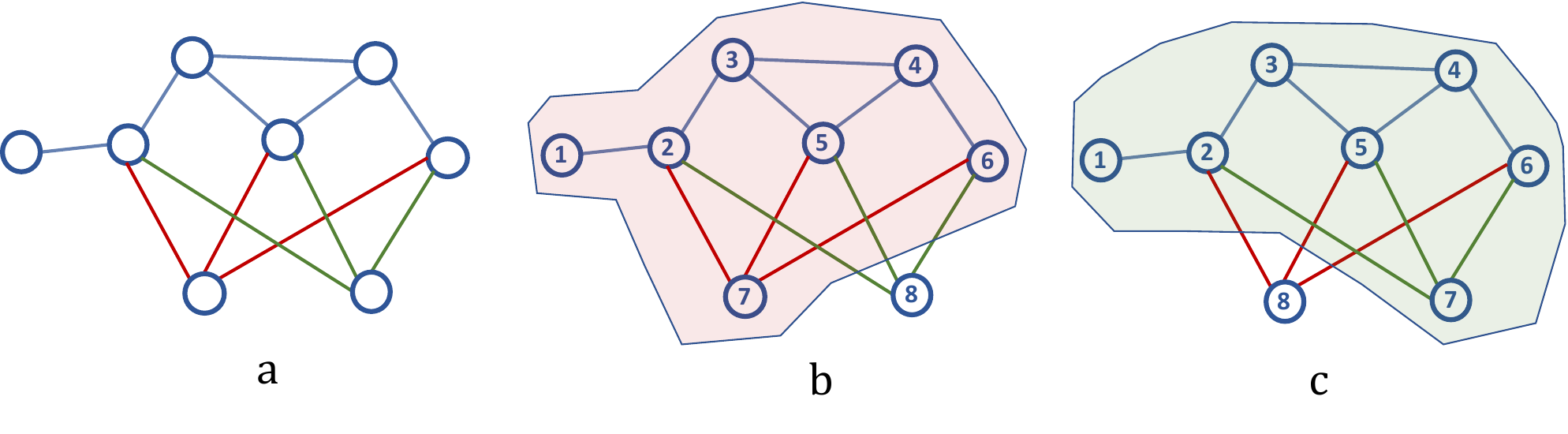}
\vspace*{-1 em}
%    \caption{Example showing how temporal clustering arises: nodes $u$ and $v$ are indistinguishable as they have the same set of neighbors $\N(u)$ and $\N(v)$. Thus unique labeling of nodes based on arrival order is not possible, and $u$ and $v$ should end up in the same temporal cluster.}
\caption{\small Example showing how temporal clustering arises: a) the input graph without labels. b) and c) arbitrary labellings of arrival order with $1$ representing the earliest arrival and $8$ for the latest arrival. In b) and c), the last two arrived nodes $7$ and $8$ have the same set of neighbors.
If we simulate the process of evolution starting from node $1$, we observe that graphs in b) and c) at time $7$ (i.e., with nodes $1-7$) are identical.
Thus, nodes $7$ and $8$ in b) and c) are indistinguishable as to which arrived early between them (this observation holds for any labeling on the input graph), and the nodes behind these two labels are part of a temporal cluster.
}
    \label{fig:intro_example}
%\vspace*{-1 em}
\end{figure}

Temporal ordered clustering is related to many applications in practice. For example in online social networks, it can be useful to disseminate specific information or advertisements targeted at nodes that arrived around the same time. In biological networks, it identifies the evolution of biomolecules in the network and helps in predicting early proteins that are known to be preferentially implicated in cancers and other diseases~\cite{srivastava2010amphimedon}. In rumor or epidemic networks, temporal ordered clustering can assist in identifying the sources and carriers of false information.

\xhdr{Our contributions}
\begin{itemize}[leftmargin= 2ex,topsep=1 ex]
\item We provide a general framework and derive an optimization problem for finding temporal ordered clusters in dynamic networks when only the final snapshot of its evolution is provided.
% Due to high computational complexity involved in solving the optimization problem of temporal ordered clustering,
Due to high computational complexity involved in solving it, we reformulate the problem in terms of partial orders -- for any node pairs $(u, v)$, a partial order $\sigma$ defines an order $u <_{\sigma} v$ in which node $u$ is specified to be arrived earlier than node $v$.
Such a partial order naturally translates into clusters of nodes and introduces an order among them.
Both the optimization problems depend on the knowledge of the probabilistic evolution of the graph model and the probability that any node $u$ is older than any other node $v$, denoted as $p_{u,v}$. We then design a {\it sequential importance sampling} algorithm to estimate $p_{u,v}$ for any general graph model, and prove its convergence. The solution to a linear programming relaxation of the original optimization problem, with coefficients as estimated $p_{u,v}$, presents an upper bound on the clustering quality.

\item In the case of large networks, when the complexity for solving the original optimization is higher, approximate solutions, which directly make use of the estimated $p_{u,v}$s, are developed. Moreover, when some information about the node-pair orders are available, we develop semi-supervised techniques that exploit the graph structure to improve estimation precision. We observe that the use of semi-supervised learning enhances the estimated values of $p_{u,v}$ quickly even with a small percentage of labeled data.

\item In the second part of the paper, as an application of the proposed general technique, we focus on {\it duplication-divergence} or vertex copying dynamic network model (DD-model) in which, informally, a new node copies the edges of a randomly selected existing node and retains them with a certain probability, and also makes random connections to the remaining nodes (see Section~\ref{subsec:dd-model} for details).
%\KT{If the node copies some edges and not all, it's redundant to add that retains them. I'd suggest "in which, informally, a new node copies randomly some edges of a randomly selected existing node and also makes random connections to the remaining nodes"}
The DD-model poses unique challenges for temporal ordered clustering in comparison with other graph models because of the features listed below:
\begin{itemize}
	\item {\em Non-equiprobable large number of permutations}: In many of the graph models including the preferential attachment and Erd\H{o}s-R\'enyi graph models, all the feasible permutations of the same structure representing node arrival orders are equally likely \cite{luczak2016asymmetry}. Later in the paper, we show with a counter example that this is not the case in the DD-model.
	In other words, unlike in our previous work \cite{sreedharan2019inferring}, we do not assume the isomorphic graphs that have positive probability under the graph model have the same probability.
	Moreover, in the DD-model, all the permutations of node labels with $n$ letters are valid unlike some models like preferential attachment model and hence the effective space of total orderings is $n!$. Thus the DD-model stands as corner case in the problem of node arrival order inference.

	\item {\em Large number of symmetry}: We provide evidence of a large number of automorphisms in a duplication-divergence graph, whereas it is known that  Erd\H{o}s-R\'enyi and preferential attachment graphs are {\em asymmetric} (when the automorphism group contains only the identity permutation) with high probability~\cite{kim2002asymmetry,luczak2016asymmetry}.

	\item {\em Ineffectiveness of degree-based techniques}: In some models (including preferential attachment model), the oldest nodes have larger expected degrees than the youngest nodes over time, with high probability. But it is known that in the DD-model the average degree does not exhibit such a consistent trend \cite{psexpected,frieze2020}. Thus any method based on degrees is bound to fail in the DD-model.
\end{itemize}
\end{itemize}

\xhdr{Prior related work}
Graph clustering is a well-studied problem which, in general, follows two main approaches: 1) define a similarity metric between node pairs, and choose clusters so as to maximize similarity among the nodes inside a cluster and minimize similarity between nodes in different clusters; 2) identify subgraphs within the input graph that reach a certain value of fitness measure, usually based on subgraph density, conductance, normalized cut or sparse cut \cite{SCHAEFFER200727}.
Many of the clustering techniques on static graphs have been extended to dynamic graphs, where primarily the aim was to study the evolution of fitness or similarity based clusters~\cite{loukas2018spectrally,liu2018global,gorke2010modularity, greene2010tracking}.

The temporal ordered clustering or partial order inference considered in this paper poses a very different problem in contrast to the classical formulation.
The optimization criterion for temporal ordered clustering introduces a fresh look taking into account the graph model and its temporal behavior (see Section~\ref{sec:formln}).
% Incidentally, the main aim of our clustering is to characterize the inability of recovering the history of a dynamic network due to the symmetries exhibited by the associated model.
The main aim of our clustering formulation is to characterize the inherent limits and to develop estimation algorithms for recovering the history of a dynamic network.
The nodes inside our clusters are indistinguishable in terms of their arrival order due to symmetries in the input graph and there exists a hierarchy or order among the clusters with respect to graph evolution.

Previous works on semi-supervised clustering methods for data represented as vectors \cite{bair2013semi, basu2002semi} and their extensions to graphs \cite{kulis2009semi} focus mainly on using the labeled nodes to define clusters and their centroids. However, in temporal ordered clustering, the labeled nodes need not fully represent all the clusters, and they are used to reduce the complexity of estimation of coefficients of the associated linear programming (by restricting the sampling distribution of importance sampling, see Section~\ref{subsec:sup_soln})

Node arrival order in the DD-model has been studied in \cite{li2013maximum} and \cite{navlakha2011network}, and the references therein.
% Most of the prior works focus on getting the complete arrival order of nodes (total order), for which we show later that for large graphs, with most of the parameter settings of the underlying graph model, total order is not better than random guessing, and does not result in unique order due to the inherent symmetries.
Most of the prior works focus on getting the complete arrival order of nodes (total order), but it turns out that it becomes nearly impossible due to their symmetries \cite{luczak2016asymmetry,turowski2018allerton}.
Instead of total order, in this work we focus on deriving an optimal partial order of nodes of nodes (see Section~\ref{sec:prob_formulation}).
Our methods are general and are applicable to a wide class of graph models, unlike our recent work \cite{sreedharan2019inferring} where the methods were specific to the preferential attachment model and not extendable.

A preliminary version of this paper is appeared in \cite{turowski2020temporal}.

% \xhdr{Main notation}
% In the following $\bpi$ always represents a uniform random permutation from $S_n$, the symmetric group on $n$ letters and $n$ will be implied from the context. Let $G_n, \mathcal{G}_n, \mathbb{G}_n$ be the deterministic graph, random graph from the model under consideration and the set of graphs with $n$ nodes. Similar definition holds for $H_n$ and $\mathcal{H}_n$, but with $\mathcal{H}_n = \bpi(\mathcal{G}_n)$. We label the vertices in the original graph $G$ in their arrival order, $[n] = \{1,...,n\}$, where node $j$ is the $j$th node to arrive. For a Markov chain $\{X_k\}_{k \geq 0}$ with transition probability matrix $P$, $\E_x[\cdot]$ indicate expectation with respect to its sample paths starting from $X_0=x$.
%The variables $\theta$ and $\delta$ denote precision and density of the estimator, $\alpha$ is the proportion of the training set used in supervised learning, and $\epsilon$ denote the proportion of ${n \choose 2}$.

\section{Problem Formulation}
\label{sec:prob_formulation}
Let $H_n$ be the observed undirected and unweighted graph of $n$ nodes with $V(H_n)$ being the set of vertices and $E(H_n)$ being the set of edges. The graph $H_n$ is a result of evolution over time, starting from a seed graph $H_{n_0}$ with $n_0$ nodes. At a time instant $k$, when a new node appears, a set of new edges adjacent to the new node is added, and the graph $H_k$ will evolve into $H_{k+1}$.
Since the change in graph structure occurs only when a new node is added, assuming the addition of a node as a time epoch, $H_n$ also represents graph at time epoch $n$. The time epoch $n_0$ denotes the creation of the seed graph $G_{n_0}$\footnote{In the rest of the paper, we omit conditioning on the given $G_{n_0}$ in all the expressions for the sake of brevity, if it is clear from the context.}.

Given only the snapshot of the dynamic graph $H_n$ at time $n$, we usually do not know the time or order of arrivals of nodes. Essentially, our goal is to label each node with a number $i$, $1 \leq i \leq K$, such that all the nodes labeled by $i$ arrived before nodes with labels $j$ where $j> i$. The number of labels (clusters) $K$ is unknown before and is a part of the optimal clustering formulation.
The arrival of a new node and the strategy it uses to choose the existing nodes to make connections depend on the graph generation model. We thus express the above problem in the following way.
Let $G_n$ be a graph drawn from a dynamic random graph model $\cG_n$ on $n$ vertices in which
nodes are labeled as $[n] = \{1, 2, \ldots, n\}$ according to their arrival, i.e., node
$j$ was the $j$th node to arrive. Let $G_n$ evolve from the seed graph $G_{n_0}$.
To model the lack of knowledge of the original labels, we subject
the nodes to a permutation $\bpi$ drawn uniformly at random from the symmetric group on $n$ letters $S_n$, and we are given the graph $H_n:=\bpi(G_n)$; that is, the nodes of $G_n$ are randomly \emph{relabeled}. We also use the notation $\cH_n$ to denote the random graph behind $H_n$.
Our original goal is to infer the arrival order in $G_n$ after
observing $H_n$, i.e., to find $\bpi^{-1}$. The permutation
$\bpi^{-1}$ gives the true arrival order
of the nodes of the given graph.

Instead of putting a constraint on recovering the whole permutation $\bpi^{-1}$ or equivalently $K=n$ labels, we resort to strict (irreflexive) partial orders.
For a partial order $\sigma$, a relation $u <_\sigma v$ means that node $u$ is older than node $v$ according to the ordering $\sigma$..
%\KT{First part sounds redundant -- "a relation $u <_\sigma v$ means that node $u$ is older than node $v$ according to the ordering $\sigma$."}

\subsection{Relation between temporal ordered clusters and partial order set}
\label{subsec:relation_cluster_partial_order}
Every partially ordered set can be represented by a clustering $\{\C_i\}$ as follows. A strict partially ordered set can be represented initially by a directed acyclic graph (DAG) with nodes as the nodes in the graph $H_n$ and directed edges as given by the partial order $\sigma$: an edge from $v$ to $u$ exists when $u <_\sigma v$. Then taking the transitive closure of this DAG will result in the DAG of the partial order set $\sigma$. Now, all the nodes with in-degree $0$ in the DAG will be part of cluster $C_K$ and the set of nodes with all the in-edges coming from nodes in $C_{K}$ will form cluster $C_{K-1}$. This process repeats until we get $C_1$. The number of clusters $K$ is not defined before but found from the DAG structure. Unlike the classical clustering, these clusters are ordered such that $\C_1 \prec \C_2 \ldots \prec \C_K$, where the relation $\C_i \prec \C_j, i<j$ is defined as all the nodes inside the cluster $C_i$ are estimated to be arrived earlier than all the nodes in the cluster $\C_j$, and all the nodes inside each cluster are considered to be identical in arrival order.
We note here that not all partial orders result in a DAG that is weakly connected.
If there are multiple components in the DAG corresponding to a partial order, each of them will give independent clustering. It might be due to the nodes in these separate components of the DAG are developed independently during evolution.
Moreover, if there are nodes that are not part of any comparison in the partial order, we label them as unclassified.

In the following Section~\ref{subsec:optzn_cluster}, we formulate an optimization problem for the clusters and find that the time complexity of its solution is $n^5$-times larger than that of the solution of the optimization problem of partial orders in Section~\ref{subsec:optzn_partial_orders}. Hence in this paper, we focus only on the temporal-ordered clusters derived from the partial order.

We define an estimator $\phi$ of the temporal ordered clustering\footnote{From now on, we use the terms node arrival order inferencing and temporal ordered clustering interchangeably in the paper} as a function $\phi$ from the set of all labeled graphs on $n$ vertices to the set of all partial orders on nodes $1,\ldots,n$.

We consider estimators based on {\em unsupervised} and {\em semi-supervised} learning paradigms:
\begin{itemize}[leftmargin= 2ex,topsep=0 ex]
\item Unsupervised: In this case, the estimator does not have access to any information of the node arrival orders. Its results will be based only on the assumption that the graph model fits well the real-world network under consideration. In Section~\ref{sec:formln} we formulate an optimization problem for unsupervised learning and in Section~\ref{sec:approx_opt_soln} we provide approximate solutions of the optimization.

\item Semi-supervised: In some of the real-world networks, partial information of the order of nodes is available -  for some of the node pairs $u,v$, it is revealed to the estimator that node $u$ is arrived earlier than node $v$. Such node pairs are termed as {\em perfect pairs}. Taking this information into account would help the estimator that is initially based on fixed graph model to adapt to the real-data. The semi-supervised estimators introduced in Section~\ref{sec:approx_opt_soln} learn the partial orders in the data without violating the perfect pairs.
\end{itemize}

\subsection{Measures for evaluating partial order}
For a partial
order $\sigma$, let $K(\sigma)$ denote the number of pairs $(u, v)$
that are comparable under $\sigma$: i.e.,
$K(\sigma) = | \{ (u, v) ~:~ u <_{\sigma} v \}|$,
where $|K(\sigma)| \leq \binom{n}{2}$.
\parsec
\emph{Density}: the density of a partial order $\sigma$ is simply the number
of comparable pairs, normalized by the total possible number, $\binom{n}{2}$.  That is,
%\[
$\delta(\sigma)
= {K(\sigma)}/{\binom{n}{2}}. %|\{ (u, v) \in [n] ~:~ u <_{\sigma} v\} |.
$ %\]
Note that $\delta(\sigma) \in [0, 1]$.  Then the density of a partial order estimator
$\phi$ is simply its minimum possible density $\delta(\phi)
= \min_{H_n}[\delta(\phi(H_n))]$.

\vspace{0.25em}
\noindent \emph{Precision}: it measures the
expected fraction of \emph{correct} pairs out of all pairs that
are guessed by the partial order. That is
\[
\theta(\sigma)
= \E \left[
\frac{1}{K(\sigma)} | \{ u, v \in [n]\colon u <_{\sigma} v, \bpi^{-1}(u) < \bpi^{-1}(v) \} |
\right].
\]
For an estimator $\phi$, we also denote by $\theta(\phi)$ the quantity $\displaystyle\E[\theta(\phi(\bpi(\cG_n)))]$.
We note here that the typical graph clustering performance measures like Silhouette index and Davies-Bouldin index do not find useful in our set up since the distance measure in our case is difficult to capture quantitatively and is purely based on indistinguishability due to symmetries and arrival order of nodes.

\section{Solving the Optimization Problem}
\label{sec:formln}
The precision of a given estimator $\phi$ can be written in the form of a sum over all graphs $H_n$:
\begin{flalign*}
&\theta(\phi)
= \sum_{H_n} \Pr[\bpi(\cG_n)=H_n] \frac{1}{K(\phi(H_n))} \\
& \times \E \left[ |\{ u,v\in [n] \colon u <_{\phi(H_n)} v, \bpi^{-1}(u) < \bpi^{-1}(v) \}| \Big| \bpi(\cG_n)=H_n \right]\!.
\end{flalign*}
%\KT{If $u, v \in [n]$, then $\cG$ and $H$ should be $\cG_n$ and $H_n$, respectively.}
Here $\bpi$ and $\cG_n$ are the random quantities in the conditional expectation.
We formulate the optimal estimator as the one that gives maximum precision for a given minimum density.
For an estimator to be optimal, it is then sufficient to choose,
for each $H_n$, a partial order $\phi(H_n)$ that maximizes
%\KT{Maybe "it is sufficient for each $H$ to choose a partial order"?}
\begin{align*}
&J_{\epsilon}(\phi) := {K(\phi(H_n))}^{-1} \\
& \times {\E \left[ |\{ u,v\in [n] \colon u <_{\phi(H_n)} v, \bpi^{-1}(u) < \bpi^{-1}(v) \}|\Big| \bpi(\cG_n)=H_n \right]\!.} \nn
\end{align*}
subject to the density constraint $\delta(\phi(H_n)) = K(\phi(H_n))/\binom{n}{2} \ge \varepsilon$, which says that we must have a certain minimum density of comparable pairs (here, $\epsilon \in [0, 1]$ is a parameter of the problem).
% $K(\phi(H)) \geq \epsilon \binom{n}{2}$.
%\KT{I'd write there explicitly "$\delta(\phi(H)) = K(\phi(H))/\binom{n}{2} \ge \varepsilon$".}

In the the following first two subsections, we formulate the above optimization problem for two cases: when the estimator outputs the clusters and when it outputs the partial order. Each of these optimizations add a set of extra constraints to the original problem.

Let
\begin{align}
p_{u,v}(H_n) & := \Pr[\bpi^{-1}(u) < \bpi^{-1}(v) | \bpi(\cG_n)=H_n]
% \nn \\
% {P}_{\text{pair}}(H_n) & := \{p_{u,v}(H_n), \forall u,v \in H_n\}. \nn
\end{align}
be the probability that $u$ is arrived before $v$ given the relabeled graph $H_n$.
The probability $p_{u,v}(H_n)$ turns out to be a critical quantity that serves as the coefficient in the linear programming approximations of the optimization problems and its estimation is explained in the last subsection of this section\footnote{We drop the dependence of $H_n$ in $p_{u,v}(H_n)$ and $P$ if it is clear from the context.}.

\subsection{Integer programming formulation for clusters}
\label{subsec:optzn_cluster}
In this subsection, we restrict our optimization to linear cluster estimators, where the clusters are arranged in a total (linear) order.

To accomplish this optimization, we introduce, for each vertex $v$, a vector $\vec{x}_v = (x_{v,1}, \ldots, x_{v,n})$, where $x_{v,i} = 1$ encodes the fact that node $v$ is placed in cluster $i$.

Then $J_{\epsilon}$ can be written in terms of integer programming (IP) formulation as
\begin{align}
% &\E \left[ \sum_{1 \leq u \neq v\leq n} \mathbf{1}\left\{ u <_{\phi(H_n)} v, \bpi^{-1}(u) < \bpi^{-1}(v) \right\} \Big| \bpi(\cG_n)=H_n \right] \nn \\
\sum_{1 \leq u \neq v\leq n} \sum_{1 \leq i < j  \leq n} p_{u,v}(H_n)\,  \dfrac{x_{u,i} x_{v,j}}{\displaystyle\sum_{1 \leq k < l \leq n} \sum_{1 \leq w \neq w' \leq n} x_{w,k} x_{w',l}},
\label{eq:IP-cluster}
\end{align}
subject to the basic constraints\footnote{Let the nodes in $H_n$ take unique labels from the set $[n]= \{1,2,\ldots,n \}$ (the original random graph $\cG_n$ is assumed to be labeled from $[n]$, with label $i$ indicating $i$th arrival node).}
\[ \sum_{j=1}^n x_{v,j} =1,  \forall j \in [n] \quad \& \quad  x_{v, j} \in \{0,1\}, \forall v \in [n], \forall j \in [n].\]

\noindent We additionally have the following density constraint for a given $\epsilon$:
\[ \sum_{1 \leq k < l \leq n} \sum_{1 \leq w \neq w' \leq n} x_{w,k} x_{w',l} \geq \epsilon \binom{n}{2}. \]
Each term of the form $x_{u,i} x_{v,j}$ becomes one only when the node $u$ is classified into a cluster $i$ that has lower precedence than node $v$'s cluster $j$ ($i < j$). This corresponds to the event $u <_{\phi(H_n)} v$ with $\phi$ as given by the clusters. The probability $p_{u,v}$ appears because of the event $\bpi^{-1}(u) < \bpi^{-1}(v)$ inside the expectation in $J_{\epsilon}$. The denominator in \eqref{eq:IP-cluster} corresponds to $K(\phi(H_n))$.

That is, we have a quadratic rational integer program with linear basic constraints and a quadratic constraint introduced by the minimum density. We show now how to convert our program to a linear rational integer program with linear constraints.

We define new variables $z_{u,i,v,j} = x_{u,i} x_{v,j}$ , for $u, v, i, j \in [n]$.
We can then eliminate the rational part of the integer program using the substitution
\begin{gather*}
  s = \left(\sum_{\substack{1 \leq k < l \leq n \\ 1 \leq w \neq w' \leq n}} z_{w,k, w',l}\right)^{-1} \text{ and } z'_{u,i,v,j} = s \, z_{u,i,v,j}.
\end{gather*}
With the above change of variables, the domain of $z'$ is restricted to $\{0,s\}$. The density constraint
\[\sum_{\substack{1 \leq u \neq v\leq n \\ 1 \leq i < j  \leq n}} z_{u,i,v,j} \geq \epsilon \binom{n}{2} \implies s \leq \frac{1}{\epsilon \binom{n}{2}}.\]
Now we transform the integer program to a linear program by assuming $z'$ takes continuous values with domain $[0,1/\epsilon \binom{n}{2}]$.
We call the resulting optimization as {\bf LP-clusters}.

\begin{table}[!ht]
\small
\vspace{0 em}
\renewcommand{\arraystretch}{0.5}
\begin{tabular}{b{0.21\textwidth}|b{0.24\textwidth}}
\multicolumn{1}{c}{\textbf{Original integer program}} \vspace{1 em}
 & \multicolumn{1}{c}{\textbf{LP approximation}} \\
\multicolumn{1}{c|}{\hspace{-0.1 em}$\max\limits_{z}\frac{\displaystyle\sum\limits_{\substack{1 \leq u \neq v\leq n \\ 1 \leq i < j  \leq n}} p_{u,v}(H_n)\, z_{u,i,v,j}}{\displaystyle\sum\limits_{\substack{1 \leq k < l \leq n \\ 1 \leq w \neq w' \leq n}} z_{w,k, w',l}}$} & \multicolumn{1}{c}{\hspace{-1.8 em}
$\max\limits_{z'} \displaystyle\sum\limits_{\substack{1 \leq u \neq v\leq n \\ 1 \leq i < j  \leq n}} p_{u,v}(H_n) z'_{u,i,v,j}$} \\
subject to & subject to \\
% \, $\sbullet[0.75]~ x_{u,i} \in \{0, 1\}$, $\forall u, v \in [n]$
% & \, $\sbullet[0.75]~ x'_{u,v} \in [0, s]$, $\forall u,v \in [n]$\\
\, $\sbullet[0.75]~ z_{u,i,v,j} \in \{0, 1\}$
& \, $\sbullet[0.75]~ z'_{u,i,v,j} \in [0, 1/\epsilon \binom{n}{2}]$\\
\multicolumn{1}{r|}{ $\forall u,i,v,j \in [n]$} & \multicolumn{1}{r}{$\forall u, i, v, j \in [n]$} \\
\, $\sbullet[0.75]~\displaystyle \sum_{\substack{1 \leq u \neq v\leq n \\ 1 \leq i < j  \leq n}} z_{u,i,v,j} \geq \epsilon \binom{n}{2}$
& \, $\sbullet[0.75]~\displaystyle\sum_{\substack{1 \leq u \neq v\leq n \\ 1 \leq i < j  \leq n}} z'_{u,i,v,j} = 1$\\
% \, $\sbullet[0.75]~\displaystyle\sum_{i \in [n]} x_{u, i} = 1$, $\forall u \in [n]$
% & \, $\sbullet[0.75]~\displaystyle\sum_{i \in [n]} x_{u, i} = s$, $\forall u \in [n]$\\
% \, $\sbullet[0.75]~z_{u,i, u,i} = x_{u,i}$, $\forall u, i \in [n]$
% & \, $\sbullet[0.75]~z'_{u,i, u,i} = x'_{u,i}$,  $\forall u, i \in [n]$\\
\, $\sbullet[0.75]~\displaystyle \sum_{i \in [n]}z_{u,i, u,i} = 1$, $\forall u \in [n]$
& \, $\sbullet[0.75]~ \displaystyle\sum_{i \in [n]}z'_{u,i, u,i} \leq 1/\epsilon \binom{n}{2}$,  $\forall u \in [n]$ \\
\, $\sbullet[0.75]~z_{u,i,v,j} = z_{v,j,u,i}$
& \, $\sbullet[0.75]~z'_{u,i,v,j} = z'_{v,j,u,i}$ \\
\multicolumn{1}{r|}{ $\forall u,i,v,j \in [n]$} & \multicolumn{1}{r}{$\forall u,i,v,j \in [n]$} \\
% \, $\sbullet[0.75]~\displaystyle\sum_{i \in [n]} z_{u,i,v,j} = x_{v,j}$,
% & \, $\sbullet[0.75]~\displaystyle\sum_{i \in [n]} z'_{u,i,v,j} = x'_{v,j}$, \\
% \multicolumn{1}{r|}{$\forall u,v,j \in [n]$.} & \multicolumn{1}{r}{$\forall u,v,j  \in [n]$.}\\
\, $\sbullet[0.75]~\displaystyle\sum_{i \in [n]} z_{u,i,v,j} = z_{v,j,v,j}$,
& \, $\sbullet[0.75]~\displaystyle\sum_{i \in [n]} z'_{u,i,v,j} = z'_{v,j,v,j}$, \\
\multicolumn{1}{r|}{$\forall u,v,j \in [n]$} & \multicolumn{1}{r}{$\forall u,v,j  \in [n]$}
\end{tabular}
% \vspace{-2 em}
\end{table}

The first three constraints are direct translation of the constraints in \eqref{eq:IP-cluster}, and the last two comes from the substitution $z_{u,i,v,j} = x_{u,i} x_{v,j}$.

\xhdr{Complexity analysis}
The LP approximation presented above has $\Theta(n^4)$ decision variables and $\Theta(n^4)$ constraints.
Therefore the complexity of solving this optimization, without taking into account the complexity of estimating $p_{u,v}$, will be of the order of $n^{12}$ ($d^2c$ if $d$ is the number of decision variables and $c$ is the number of constraints \cite[Section 1.2.2]{boyd2004convex}).

The numerical experiments of the optimization in terms of clusters are presented later in Section~\ref{subsec:synthetic_data_numerical}. We also provide comparisons showing the formulation in terms of partial orders given in the next subsection computes much faster, yet outputs estimates with precision closer to that of cluster optimization.

\subsection{Integer programming formulation for partial orders}
\label{subsec:optzn_partial_orders}
In this subsection, we derive the optimal partial order among the nodes for the arrival order inference problem, extending some results from our recent work in \cite{sreedharan2019inferring}.

We now represent the optimization problem with $J_{\epsilon}(\phi)$ as an integer program of partial order.
For an estimator $\phi$, we define a binary variable $y_{u,v}$ for each ordered pair $(u,v)$ as $y_{u,v}=1$ when $u <_{\phi(H_n)} v$. Note that $y_{u,v}=0$ means either $u >_{\phi(H_n)} v$ or the pair $(u,v)$ is incomparable in the partial order $\phi(H_n)$.

%Then $J_\epsilon(\phi)$ can be expressed as
%\begin{align}
%J_{\epsilon}(\phi)
%= \frac{\sum_{1 \leq u \neq v \leq n } p_{u,v}(H) x_{u,v} }{ \sum_{1 \leq u \neq v \leq n} x_{u,v}}. \nn
%\end{align}
%\begin{enumerate}
%    \item Antisymmetry: $x_{u,v} + x_{v,u} \leq 1$.
%    \item Transitivity: $x_{u,w} \geq x_{u,v} + x_{v,w} - 1$ for all $u, v, w \in [n]$.
%    \item Minimum density: $\sum_{1 \leq u \neq v \leq n} x_{u,v} \geq \epsilon \binom{n}{2}$.
%    \item Domain restriction: $x_{u,v} \in \{0, 1\}$ for all $u, v \in [n]$.
%\end{enumerate}
%\noindent \begin{minipage}[t]{.4\linewidth}
%    \begin{align*}
%    \max_{y,s} \sum_{1 \le u \neq v \le n} p_{u,v}(H) x'_{u,v}
%    \label{eq:opt_lp}
%    \end{align*}
%\end{minipage}
%\begin{minipage}[t]{.5\linewidth}
%    \begin{align*}
%    \text{s.t}~~ & x'_{u,v} \in [0,s], \forall u,v \in [n] \\
%    & x'_{u,v} + x'_{v,u} \leq s, \forall u,v  \in [n] \\
%    & x'_{u,v} + x'_{v,w} - x'_{u,w} \leq s, \forall u,v,w  \in [n] \\
%    & \sum_{1 \le u \neq v \le n} x'_{u,v} = 1.
%    \end{align*}
%\end{minipage}

In the following, we write the optimization in two forms: the original integer program (left) and the linear programming approximation (right). The objective functions of both the formulations are equivalent to $J_\epsilon(\phi)$.
The constraints of the optimizations correspond to domain restriction, minimum density, and partial order constraints --  antisymmetry and transitivity respectively.
To use a linear programming approximation, we first convert the rational integer program into an equivalent truly integer program.
With the substitution $s = 1/\sum_{1 \leq u \neq v \leq n} y_{u,v}$, and $y'_{u,v} = s y_{u,v}$, the objective function is rewritten as a linear function of the normalized variables. These programs are equivalent if $y'_{u, v} \in \{0, s\}$, $s\leq {1}/{\varepsilon \binom{n}{2}}$. For the LP relaxation, we assume $y'_{u,v}$ as $\left[0,{1}/{\varepsilon \binom{n}{2}}\right]$. We call the LP in this subsection as the {\bf LP-partial-order}.

\begin{table}[ht]
\small
\vspace{0 em}
\begin{tabular}{b{0.21\textwidth}|b{0.23\textwidth}}
\multicolumn{1}{c}{\textbf{Original integer program}} \vspace{1 em}
 & \multicolumn{1}{c}{\textbf{LP approximation}} \\
\multicolumn{1}{c|}{ \hspace{-0.1 em}$\displaystyle \max_{y} \frac{\sum_{1 \leq u \neq v \leq n } p_{u,v}(H_n) y_{u,v} }{\sum_{1 \leq u \neq v \leq n} y_{u,v}}$} & \multicolumn{1}{c}{\hspace{-2.5 em}
$\displaystyle \max_{y'} \sum_{1 \le u \neq v \le n} p_{u,v}(H_n) y'_{u,v}$} \\
subject to & subject to \\
\Hquad $\sbullet[0.75]~ y_{u,v} \in \{0, 1\}$, $\forall u, v \in [n]$ & \Hquad $\sbullet[0.75]~y'_{u,v} \in [0,1/\epsilon \binom{n}{2}]$, $\forall u,v \in [n]$ \\
\Hquad $\sbullet[0.75]~\displaystyle \sum _{1 \leq u \neq v \leq n} y_{u,v} \geq \epsilon \binom{n}{2}$ & \Hquad $\sbullet[0.75]~\displaystyle \sum_{1 \le u \neq v \le n} y'_{u,v} = 1$ \\
\Hquad $\sbullet[0.75]~y_{u,v} + y_{v,u} \leq 1$, $\forall u,v \in [n]$ & \Hquad $\sbullet[0.75]~y'_{u,v} + y'_{v,u} \leq 1/\epsilon \binom{n}{2}$,  \\
\multicolumn{1}{r|}{ } & \multicolumn{1}{r}{$\forall u,v\in [n]$}\\
\Hquad $\sbullet[0.75]~y_{u,v} + y_{v,w} - y_{u,w} \le 1$, & \Hquad $\sbullet[0.75]~y'_{u,v} + y'_{v,w} - y'_{u,w} \leq 1/\epsilon \binom{n}{2}$, \\
\multicolumn{1}{r|}{$\forall u, v, w \in [n]$} & \multicolumn{1}{r}{$\forall u,v,w  \in [n]$}
\end{tabular}
\end{table}

The above integer program and LP-partial-order formulation is different from the LP-clusters in many ways. The idea of LP-partial-order is to relax the formulation of LP-clusters by focusing on the underlying partial order of clusters, rather than clusters itself. This simplifies the objective function, though it brings additional partial order constraints into the optimization. After finding the optimal partial order, we can derive the ordered clusters from it using the peeling technique in Section~\ref{subsec:relation_cluster_partial_order}. We note here that this may not need result in unique clusters. Many partial orders can have the same the cluster structure, especially when the DAG corresponding to the partial order contains multiple components.
% then the technique in Section~\ref{subsec:relation_cluster_partial_order} can not result in strictly ordered clusters unlike the LP-clusters.

The next lemma bounds the effect of approximating the coefficients $p_{u,v}$ on the optimal value of the integer program.
\begin{lemma}
    Consider the integer program whose objective function is given by
    \[\hat{J}_{\epsilon,\lambda}(\phi) = \frac{\sum_{1 \leq u < v \leq n}
        \hat{p}_{u,v}(H_n) y_{u,v} }{\sum_{1\leq u \neq v \leq n} y_{u,v}},\]
    with the same constraints as in the original integer program.  Assume $p_{u,v}(H_n)$ can be approximated with $|\hat{p}_{u,v}(H_n)- p_{u,v}(H_n)| \leq \lambda$ uniformly for all $u,v$.
    Let $\phi_*$ and $\hat{\phi}_*$
    denote optimal points for the original and modified integer programs, respectively.
    Then
    $ %\begin{align}
    |\hat{J}_{\epsilon,\lambda}(\hat{\phi}_*) - J_{\epsilon}(\phi_*)|
    \leq 3\lambda,
    $ %\end{align}
    for arbitrary $\lambda > 0$.
    \label{lemma:effect_perturbation_puv}
\end{lemma}
The proof of the above lemma is an extension of \cite[Lemma 5.1, Supplementary Material]{sreedharan2019inferring} -- we require a weaker assumption $|\hat{p}_{u,v}(H_n) - p_{u,v}(H_n)| \leq \lambda$ instead of $|\hat{p}_{u,v}(H_n)/p_{u,v}(H_n) - 1| \leq \lambda$ in \cite{sreedharan2019inferring}.
% the authors assume $|\hat{p}_{u,v}(H)/p_{u,v}(H) - 1| \leq \lambda$, which is a stronger condition than our present assumption $|\hat{p}_{u,v}(H) - p_{u,v}(H)| \leq \lambda$.
%\KT{"The proof of the above lemma goes along exactly the lines of the proof of Lem XYZ in \cite{sreedharan2019inferring}. There is only one difference: we require weaker assumption $|\hat{p}_{u,v}(H) - p_{u,v}(H)| \leq \lambda$, instead of $|\hat{p}_{u,v}(H)/p_{u,v}(H) - 1| \leq \lambda$.". I also changed $p_{u,v}$ to $p_{u,v}(H)$, as the first one is not defined yet.}
%
% \begin{remark}
% An equivalent formulation when $\epsilon = 1$: When we consider a total order, the optimization can be rewritten as
% \[ \max_{\sigma \in S_n} \sum_{i=1}^n \sum_{j=i+1}^n p_{\sigma(i) \sigma(j)}. \]
% \end{remark}

\xhdr{Complexity analysis and advantage over the cluster optimization} The LP approximation has $\Theta(n^2)$ decision variables and $\Theta(n^3)$ constraints (in the order of their appearance in the formulation). Thus computational complexity of the LP will of the order of $n^7$ (without taking into account the estimation complexity of $p_{u,v}$), which is much less than $n^{12}$ complexity of cluster optimization in the previous subsection. Later in Section~\ref{subsec:synthetic_data_numerical}, we provide numerical comparisons showing the formulation in terms of partial orders computes much faster, yet outputs estimates with precision closer to that of cluster optimization.

\subsection{Estimating coefficients using importance sampling}
We now discuss the importance sampling approach to estimate the coefficient $p_{u,v}$ that is needed to solve the optimization problem.
%\KT{We introduce $p_{u,v}$ here but define it a paragraph later. Shouldn't we move it here?}
The following approach to estimate $p_{u,v}$ is applicable to any general graph model with Markovian evolution (conditioned on the present state of the graph, the new state is independent of the past state).
% unlike our previous work in \cite{sreedharan2019inferring} which is specific to preferential attachment graphs.

To estimate $p_{u,v}$, we classify dynamic graph models into two categories. Let $\Gamma(H_n)$ be the set of all feasible permutations $\sigma$ which generates a positive probability graph $\sigma(H_n)$ according to the distribution of the graph generation model.
\begin{enumerate}[leftmargin= 2ex,topsep=1 ex]
    \item[i).] {\bf Graph models with equiprobable isomorphic graphs.} Here, two isomorphic graphs have same probability under the graph model. Formally, consider a graph $G_n^{(1)}$ with $\P[\cG_n=G_n^{(1)}] > 0$ and another graph $G_n^{(2)}$, $G_n^{(2)} = \sigma(G_n^{(1)})$ with $\sigma \in \Gamma(G_n^{(1)})$, then the equiprobable condition can be stated as $\P[\cG_n=G_n^{(1)}] = \P[\cG_n=G_n^{(2)}]$.
    Our previous work in \cite{sreedharan2019inferring} focus on such a case and derives the following result.

\begin{lemma}[{\cite[Lemma~4.1 in Supplementary Information]{sreedharan2019inferring}}]
    For all $v, w \in [n]$ and graphs $H_n$,
    \begin{align}
    \MoveEqLeft{\P[\bpi^{-1}(v) < \bpi^{-1}(w)| \bpi(\cG_n) = H_n]}& \nn \\ &= \frac{|\sigma: \sigma^{-1} \in (H_n), \sigma^{-1}(v)< \sigma^{-1}(w)|}{|\Gamma(H_n)|}.
    \label{eq:lemma_equiprobable_nature_paper}
    \end{align}
\end{lemma}
Though the graph models with such a property are not common, it include preferential attachment and Erd\H{o}s-Renyi models. For preferential attachment model, we show in \cite{sreedharan2019inferring} that the estimation of right-hand side of \eqref{eq:lemma_equiprobable_nature_paper} deduces to finding the proportion of linear extensions $\sigma$ of a partial order (set of node pair orderings that hold with probability $1$) satisfying $\sigma^{-1}(v) < \sigma^{-1}(w)$.

    \item[ii).] {\bf Graph models with non-equiprobable isomorphic graphs.} Many of the graph models do not possess equiprobable ismorphic graphs property. In this work, we propose a new estimation scheme based on importance sampling that is applicable to such a case for any general graph model with Markovian evolution.
\end{enumerate}

We have, for $p_{u,v}:= \P({\bpi}^{-1}(u)<\bpi^{-1}(v)|\bpi(\cG_n)=H_n) $,
{\allowdisplaybreaks
\begin{align}
p_{u,v} & = \sum_{\substack{\sigma ~:~ \sigma^{-1} \in \Gamma(H_n) \\ \sigma^{-1}(u) < \sigma^{-1}(v)}} \P(\bpi = \sigma | \bpi(\cG_n) = H_n) \nn \\
& = \sum_{\substack{\sigma ~:~ \sigma^{-1} \in \Gamma(H_n) \\ \sigma^{-1}(u) < \sigma^{-1}(v)}} \frac{\P[\bpi = \sigma, \bpi(\cG_n) = H_n]}{\P[\bpi(\cG_n) = H_n]} \nn \\
&= \sum_{\substack{\sigma ~:~ \sigma^{-1} \in \Gamma(H_n) \\ \sigma^{-1}(u) < \sigma^{-1}(v)}} \frac{\P[\cG_n = \sigma^{-1}(H_n)] \P[\bpi = \sigma]}{\sum_{\sigma^{-1} \in \Gamma(H_n)} \P[\cG_n = \sigma^{-1}(H_n)] \P[\bpi=\sigma]} \nn \\
& = \frac{\sum_{\substack{\sigma ~:~ \sigma^{-1} \in \Gamma(H_n) \\ \sigma^{-1}(u) < \sigma^{-1}(v)}} \P[\cG_n = \sigma^{-1}(H_n)]}{\sum_{\sigma^{-1} \in \Gamma(H_n)} \P[\cG_n = \sigma^{-1}(H_n)]},
\label{eq:puv_expression}
\end{align}
}
\noindent where we used the fact that $\P[\bpi=\sigma] = 1/n!$ since it is independent of $H_n$.

We now derive an estimator for $p_{u,v}$ by approximating the numerator and denominator of right-hand side in \eqref{eq:puv_expression}.
The $p_{u,v}$ expression involves summing over permutations from the feasible set, i.e., $\sigma^{-1} \in \Gamma(H_n)$ ($\sigma^{-1}(H_n)$ gives a positive probable graph by the definition of $\Gamma(H_n)$). Since there are at most $n!$ permutations to check for feasibility, direct sampling from $\Gamma(H_n)$ is impossible in many cases.
However, we remark that each permutation $\sigma^{-1} \in \Gamma(H_n)$ invokes a chain structure when the graph has a Markovian evolution, as follows. Applying $\sigma^{-1}$ to $H_n$ is essentially relabeling of nodes in $H_n$ from $[n]$. Then starting from labeling a guess of the youngest node with $n$, by reverse engineering the Markovian evolution of the graph, to find the node with label $s <n$ we need to know only the node $s+1$ and the graph $H_{s+1}$.
Based on this observation, to estimate the denominator in right-hand side of \eqref{eq:puv_expression} we propose a {\em sequential} importance sampling strategy in Theorem~\ref{thm:Markov-chain-apprxn} that generalizes to any localized sampling distribution (probability to choose node $s$ after selecting node $s+1$) which meets a certain criteria.
This is directly extendable to estimating the numerator in \eqref{eq:puv_expression} too by putting an extra restriction to the sampled permutation.
Later Lemma~\ref{lemma:estimator_consistency} presents an estimator of $p_{u,v}$ using the technique derived in Theorem~\ref{thm:Markov-chain-apprxn}.

% in which best way to sample from $\Gamma(H_n)$ and gives a faster coto estimate the denominator in right-hand side of \eqref{eq:puv_expression} nvergence.
% Each permutation $\sigma^{-1} \in \Gamma(H_n)$ has a chain structure:
% \[ \Pr[\cG_n=\sigma_{[1:n]}^{-1}(H_n)] = \Pi_{t=n}^{n_0} \Pr[\cG_{t} = \sigma_{t}^{-1}(H_{n-1})|\sigma_{[n:t+1]}]\]
% The basic idea is to sample feasible permutation set $\Gamma(H_n)$ from the structure of $H_n$.
% But since sampling from actual graph distribution incurs huge complexity, in what follows, we provide a method to use any sampling distribution that satisfies certain constraints.

Let $\R_{H_n} \subseteq V(H_n)$ denote the set of {\em candidates} for youngest nodes at time $n$. The set $\R_{H_n}$ depends on the graph model. For example, in case of preferential attachment model, in which a new node attaches $m$ edges to the existing nodes with a probability distribution proportional to the degree of the existing node, $\R_{H_n}$ is the set of $m$-degree nodes.
We consider only permutations that do not change the initial graph $G_{n_0}$ labels.
For instance, if $G_{n_0}$ has three nodes and $G_n$ has $6$ nodes, we consider the following permutations (represented in cyclic notation): $(1)(2)(3)(456)$, $(1)(2)(3)(45)(6)$, $(1)(2)(3)(46)(5)$, $(1)(2)(3)(4)(56)$, $(1)(2)(3)(4)(5)(6)$.
Thus we define $H_{n_0}$ as $G_{n_0}$ itself. Since we assume $H_{n_0}$ is known, $p_{u,v}$ expression in \eqref{eq:puv_expression} has an additional conditioning of $H_{n_0}$.

Let $\delta(H_n,z_n)$ represent the graph in which the node $z_n \in \R_{H_n}$ is deleted from $H_n$. Then the graph sequence $\cH_n = H_n, \cH_{n-1} = \delta(H_n,z_n), \ldots, \cH_{n_0} = H_{n_0}$ forms a nonhomogeneous Markov chain -- nonhomogeneous because the state space $\{\mathbb{H}_s\}_{s\leq n}$ changes with $s$ and thus the transition probabilities too. Similarly $\cG_n, \cG_{n-1}, \ldots, \cG_{n_0}$ also make a Markov chain, and for a fixed permutation $\sigma$, $\sigma(G_n) = H_n$, both the above Markov chains have same transition probabilities. Let us also define the posterior probability of producing $H_n$ from $\delta(H_n,z_n)$ as
\begin{equation}
w(\delta(H_n,z_n),H_n) := \P[\cH_n = H_n|\cH_{n-1}=\delta(H_n,z_n)].
\label{eq:expression_posterior_w}
\end{equation}

The following theorem characterizes our estimator. For a Markov chain, let $\E_x$ denote the expectation with starting state $x$. Let $\mathbb{G}_s$ be the set of all labeled graphs on $s$ vertices.

%\begin{theorem}
%    Consider a time-nonhomogeneous Markov chain $\{\cH_s\}_{s \leq t}$ with transition probability matrices as $\{Q_s=[q_s(\tilde{H}_i, \tilde{H}_j)] \}_{s\leq t}$ with $\tilde{H}_i \in \mathbb{G}_s$ and $\tilde{H}_j \in \mathbb{G}_{s-1}$, and let $\bz_n, \bv_{t-1}, \ldots$ be the nodes removed randomly by the Markov chain to jump from $\cH_n$ to $\cH_{t-1}$, $\cH_{t-1}$ to $\cH_{t-2}$ etc, and forms $\cH_n = H_n, \cH_{t-1} = \delta(H_n,\bz_n)$ etc. Then we have
%    \[ \sum_{\sigma^{-1} \in \Gamma(H_n)} \P[\cG_n = \sigma^{-1}(H_n)|H_{0}] = \E_{\cH_n=H_n} \left[ \prod_{s \leq t}^{t_0+1} \frac{w(\delta(H_s,\bv_s),H_s)}{q_s(H_s,\delta(H_s,\bv_s))} \right]. \]
%    \label{thm:Markov-chain-apprxn}
%\end{theorem}
\begin{theorem}[Sequential importance sampling]
\label{thm:MC}
    Consider a time-nonhomogeneous Markov chain $\cH_n = H_n, \cH_{n-1} = \delta(H_n,\bz_n), \ldots$, where $\bz_n \in \R_{H_n}, \bz_{n-1} \in \R_{H_{n-1}}, \ldots etc$ be the nodes removed randomly by the Markov chain and let its transition probability matrices be $\{Q_s=[q_s(F', F'')] \}_{s\leq n}$ for any two graphs $F' \in \mathbb{G}_s$ and $F'' \in \mathbb{G}_{s-1}$.
    Then we have
    \begin{align*}
    \hspace{0.7 em}
        \sum_{\mathclap{\sigma^{-1} \in \Gamma(H_n)}} \P[\cG_n = \sigma^{-1}(H_n)|H_{n_0}]  = \E_{\cH_n=H_n} \left[ \prod_{s \leq n}^{n_0+1} \frac{w(\delta(H_s,\bz_s),H_s)}{q_s(H_s,\delta(H_s,\bz_s))}\right]\!.
    \end{align*}
    \label{thm:Markov-chain-apprxn}
\end{theorem}

\begin{proof}
Now we have the following iterative expression for the denominator of $p_{u,v}$.
\begin{align}
&p_{u,v}^{\text{denom}} (H_n, H_{n_0}) : =  \sum_{\sigma^{-1} \in \Gamma(H_n)}
\P[\cG_n = \sigma^{-1}(H_n)|H_{n_0}] \\
& \!\!\!= \sum_{z_n \in \R_{H_n}} \sum_{{\sigma^{-1} \in \Gamma(H_n)}}\! \P[\cG_n = \sigma^{-1}(H_n), \cG_{n-1} = \sigma_1^{-1}(\delta(H_n,z_n))|H_{n_0}], \nn
\end{align}
where $\sigma_1 \in S_{n-1}$ is the permutation $\sigma$ with ``$z_n$ maps to $n$" removed. Now we can rewrite the above expression as
\begin{align}
\sum_{z_n \in \R_{H_n}} \sum_{\sigma^{-1} \in \Gamma(H_n)} & \P[\cG_n = \sigma^{-1}(H_n) | \cG_{n-1} = \sigma_1^{-1}(\delta(H_n,z_n))] \nn \\
& \times \P[\cG_{n-1}=\sigma_1^{-1}(\delta(H_n,z_n))|H_{n_0}].\nn
\end{align}
Note that $\P[\cG_n = \sigma^{-1}(H_n) | \cG_{n-1} = \sigma_1^{-1}(\delta(H_n,z_n))]$ for a fixed $\sigma$ (thus $\sigma_1$) is equivalent to $w(\delta(H_n,z_n),H_n)$. Now introducing a transition probability $\{Q_s=[q_s(i,j)] \}_{s\leq n}$ for the Markov chain $\{H_s\}_{s\leq n}$, and using importance sampling,
\begin{align}
p_{u,v}^{\text{denom}}(H_n, H_{n_0})
& = \sum_{z_n \in \R_{H_n} } \frac{w(\delta(H_n,z_n),H_n)}{q_n({H_n,\delta(H_n,z_n)})} q_n({H_n,\delta(H_n,z_n)}) \nn \\
& \quad\quad \times \sum_{\mathclap{\sigma^{-1} \in \Gamma(\delta(H_n,z_n))}} \P[\cG_{n-1} = \sigma^{-1}(\delta(H_n,z_n))|H_{n_0}]. \nn \\
& = \sum_{z_n \in \R_{H_n} } \frac{w(\delta(H_n,z_n),H_n)}{q_n({H_n,\delta(H_n,z_n)})} q_n({H_n,\delta(H_n,z_n)}) \nn \\
& \quad \times p_{u,v}^{\text{denom}} (\delta(H_n,z_n), H_{n_0}), \label{eq:sigma-sum_objective}
\end{align}
with $p_{u,v}^{\text{denom}} (H_{n_0}, H_{n_0}) =1$. Here $q_n({H_n,\delta(H_n,z_n)})$ is the transition probability to jump from $\cH_n = H_n$ to $\cH_{n-1} = \delta(H_n,z_n)$.
\begin{equation}
\text{Now let } \mu(H_{n}, H_{n_0})= \E_{\cH_n = H_n} \left[ \prod_{s \leq n}^{n_0+1} \frac{w(\delta(H_s,\bz_s),H_s)}{q_s(H_s,\delta(H_s,\bz_s))}\right] \nn
\end{equation}
Then we have,
\begin{align}
%\mu(H_{t}, H_{t_0}) &= \E_{\cH_t = H_t} \left[ \E_{\cH_t = H_t} \left[ \prod_{s \leq t}^{t_0+1} \frac{w(\delta(H_s,v_s),H_s)}{q_s(H_s,\delta(H_s,v_s))} \Big| \cH_{t-1} \right]\right] \nn \\
%& = \E_{\cH_t = H_t} \left[ \frac{w(\delta(H_t,V_t),H_t)}{q_t({H_t,\delta(H_t,V_t)})}\,  \E_{\cH_{t-1}} \left[ \prod_{s \leq t-1}^{t_0+1} \frac{w(\delta(H_s,v_s),H_s)}{q_s(H_s,\delta(H_s,v_s))} \right]\right]
%\label{eq:Markov-chain-lemma_proof_step} \\
\mu(H_{n}, H_{n_0})
& = \sum_{z_n \in \R_{H_n} } \frac{w(\delta(H_n,z_n),H_n)}{q_n({H_n,\delta(H_n,z_n)})} q_n({H_n,\delta(H_n,z_n)}) \nonumber \\
& \qquad  \times \E_{\cH_{n-1}=\delta(H_b,z_n)} \left[ \prod_{s \leq n-1}^{n_0+1} \frac{w(\delta(H_s,\bz_s),H_s)}{q_s(H_s,\delta(H_s,\bz_s))} \right] \label{eq:Markov-chain-lemma_proof_step} \\
& = \sum_{z_n \in \R_{H_n} } \frac{w(\delta(H_n,z_n),H_n)}{q_n({H_n,\delta(H_n,z_n)})} q_n({H_n,\delta(H_n,z_n)}) \nn \\
& \qquad \times \mu(\delta(H_n,z_n), H_{n_0}),
\label{eq:Markov-chain-lemma_proof_last}
\end{align}
where \eqref{eq:Markov-chain-lemma_proof_step} follows from the Markov property.

Defining the function at $n_0$ as
\begin{equation}
\frac{w(\delta(H_{n_0},z_{n_0}),H_{n_0+1})}{q_{n_0}(H_{n_0+1},\delta(H_{n_0},z_{n_0}))} = 1, \text{ for any $z_{n_0}$} \nn
\label{eq:lemma-Markov-condn}
\end{equation}
we note here that the iteration \eqref{eq:Markov-chain-lemma_proof_last} of $\mu(H_{n}, H_{n_0})$ is identical to that of $p_{u,v}^{\text{denom}}$ in \eqref{eq:sigma-sum_objective}. This completes the proof.
\end{proof}

\begin{remark}
Note that unlike $q_s({H_s,\delta(H_s,z_s)})$, which is under our control to design a Markov chain, $w(\delta(H_s,z_s),H_s)$ is a well-defined fixed quantity (see \eqref{eq:w_defn}). The only constraint for the transition probability matrices $\{Q_s\}_{s \leq n}$ is that it should be chosen to be in agreement with the graph evolution such that the choices of jumps from $H_s$ to $H_{s-1}$ restricts to removing nodes from $\R_{H_s}$, and it depends on the graph model.
\end{remark}

\xhdr{$\bm p_{u,v}$~estimator} Now we can form the estimator for $p_{u,v}$ for a node pair $(u,v)$ as follows.
Let $\vec{z}^{{(k)}}$ be the vector denoting the sampled node sequence of the $k$th run of the Markov chain. It can either represent a vector notation as $\vec{z}^{(k)} = (z^{(k)}_n,z^{(k)}_{n-1},\ldots,z^{(k)}_{n_0+1})$ or take a function form $\vec{z}^{(k)}(s)$ denoting the new label of a vertex $s$ in $H_n$. We propose the following estimator and show that it has asymptotic consistency.

\begin{lemma}[Estimator and its consistency]
\label{lemma:estimator_consistency}
    Let the estimator of $p_{u,v}$, for all $u,v \in H_n$, formed from $k$ samples of the sequential importance sampling (see Theorem~\ref{thm:Markov-chain-apprxn}) be
\begin{equation}
\label{eq:MC_estimator}
\widehat{p}_{u, v}^{(k)} = \dfrac{\sum_{i=1}^{k} \bone_{\{\vec{z}^{(i)}(u)<\vec{z}^{(i)}(v)\}} \prod_{s \leq n}^{n_0+1}
    \dfrac{ w(\delta(H_s,\vec{z}^{(i)}_s),H_s)}{q_s(H_s,\delta(H_s,\vec{z}^{(i)}_s))}}{\sum_{i=1}^k \prod_{s \leq n}^{n_0+1} \dfrac{w(\delta(H_s,\vec{z}^{(i)}_s),H_s)}{q_s(H_s,\delta(H_s,\vec{z}^{(i)}_s))}}.
\end{equation}
Then $\widehat{p}_{u ,v}^{(k)} \to p_{u,v}$ a.s.\ as $k \to \infty$.
\end{lemma}
\begin{proof}
Using Theorem~\ref{thm:Markov-chain-apprxn} and based on the observation that the Markov sample paths in different runs are independent and identically distributed, the numerator and denominator in the right-hand side of \eqref{eq:MC_estimator} converge separately to that of \eqref{eq:puv_expression} by strong law of large numbers (in almost surely sense). Then by invoking continuous mapping theorem, we can prove that their ratio also converges to $p_{u,v}$ almost surely.
\end{proof}

Theorem~\ref{thm:Markov-chain-apprxn} and Lemma~\ref{lemma:estimator_consistency} provide us the flexibility and convenience to sample the permutations and estimate $p_{u,v}$ via a wide-range of sampling distributions. In the next section, we consider two such candidate distributions.
\begin{algorithm*}[!htb]
        \caption{Temporal Ordered Clustering: Semi-supervised}
        \label{alg:peeling}
        \hspace*{\algorithmicindent} \textbf{Input:} graph $H_n$, graph model $\cG_n$ description, training set of partial order $\sigma_{\text{train}}$, number of sample paths $k$\\
 \hspace*{\algorithmicindent} \textbf{Output:} Clusters $\C_1 \prec \C_2 \ldots \prec \C_K$

        \begin{algorithmic}[1]
                \Procedure{TemporalOrderedClustering}{}
                \For{$\ell$ from $1$ to $k$}
                \For{$s$ from $n$ down to $n_0$}
                \State Find $\mathcal{R}_{H_s}$ and $\NR_{H_s}$ by \eqref{eq:NR}
                \State $\R_{H_s} \gets \R_{H_s} \backslash \NR_{H_s}$
                \State Sample $z_s^{(\ell)}$ using a sampling method -- \texttt{local-unif-sampling} \eqref{eq:local-unif-sampling} or \texttt{high-prob-sampling} \eqref{eq:high-prob-sampling}
                \EndFor
                \EndFor
                \State Estimate $\hat{p}_{u,v}^{(k)}, \forall u, v \in V(H_n)$ using \eqref{eq:puv_expression}
                \State Use algorithm \texttt{sort-by-$p_{u,v}$-sum} or $p_{u,v}$-\texttt{threshold} to estimate clusters of nodes $\C_1, \C_2, \ldots , \C_K$
                \State \textbf{return} $\C_1, \C_2, \ldots , \C_K$
                \EndProcedure
        \end{algorithmic}
\label{alg:semi-supervised}
\end{algorithm*}

\section{Approximating optimal solution}
\label{sec:approx_opt_soln}
In this section, we describe our main algorithms for node arrival order recovery of a general graph model.

\xhdr{Algorithms for sampling the Markov chain} Finding the whole set of permutations and calculating the exact $p_{u,v}$ according to \eqref{eq:puv_expression} is of exponential complexity. With Theorem~\ref{thm:Markov-chain-apprxn} and eq.~\eqref{eq:MC_estimator}, we can approximate $p_{u,v}$ as the empirical average of Markov chain based sample paths. We try two different importance sampling distributions $\{Q_s\}_{s\leq n}$:
\begin{itemize}[leftmargin= 2ex,topsep=1 ex]
    \item \texttt{local-unif-sampling} with transition probabilities \begin{align}
        q_s(H_s,\delta(H_s,z_s)) = \frac{1}{|\R_{H_s}|}.
        \label{eq:local-unif-sampling}
    \end{align}
    \item \texttt{high-prob-sampling} forms the Markov chain with
    \begin{align}
        q_s(H_s,\delta(H_s,z_s)) = \frac{w(\delta(H_s,z_s),H_s)}{\sum_{u \in \R_{H_s}} w(\delta(H_s,u),H_s)}.
        \label{eq:high-prob-sampling}
    \end{align}
    The above transition probability corresponds to choosing the high probability paths.
\end{itemize}
Though the \texttt{high-prob-sampling} looks like the right approach to follow, as we show later in Section~\ref{subsec:synthetic_data_numerical}, it has much slower rate of convergence than \texttt{local-unif-sampling}. Moreover at each step $s$, without taking into account the specific graph model characteristics and using naive implementations, \texttt{high-prob-sampling} requires $O(n^2)$ computations -- $O(n)$ possibilities exist for immediate ancestor of $z_s$ in $\delta(H_s,z_s)$ which is needed for calculating the posterior probability $w$ and there are $O(n)$ possibilities for the sum in the denominator, while \texttt{local-unif-sampling} requires only $O(n)$ -- counting $|\R_{H_s}|$ by checking all the nodes. In some graph models (like the DD-model in Section~\ref{subsec:dd-model}), all the nodes in $H_s$ can be part of $\R_{H_s}$ with a positive probability, and \texttt{local-unif-sampling} will essentially become uniform sampling.

% Both the above approaches are justfiedthe \texttt{high-prob-sampling} looks like the right approach to follow in theoretical terms, as we show later in Section~\ref{subsec:synthetic_data_numerical}, it has much lower rate of convergence than \texttt{local-unif-sampling}.

The \texttt{local-unif-sampling} can be further improved with the acceptance-rejection sampling technique: at a step $t$, randomly sample a node $u$ from $V(H_t)$ (instead of sampling from $\R_{H_t}$). Then calculate the probability that the node $u$ be the youngest node in the graph. If this probability is positive, we accept $u$ as $V_t$ and if it is zero, we randomly sample again from $V(H_t)$.

Now we assume that $p_{u,v}$ are estimated for all $u$ and $v$ to propose algorithms for temporal clustering. In fact, according to Lemma~\ref{lemma:effect_perturbation_puv}, we only need to have $\max_{u,v} |\hat{p}_{u,v} - p_{u,v}| \leq \lambda$ for a small $\lambda >0$. Thus for small $p_{u,v}$, $\hat{p}_{u,v}$ can be assumed to be zero. We can then use LP-partial-order in Section~\ref{subsec:optzn_partial_orders} with the estimated $p_{u,v}$ as the coefficients. Due to the huge computational complexity associated with the LP solution, we now propose the following unsupervised and semi-supervised approximation algorithms based on the estimates of $p_{u,v}$.

\subsection{Unsupervised solution}
\xhdr{\texttt{sort-by-$p_{u,v}$-sum} algorithm}
For this algorithm, we construct a new complete graph with the node set same as that of $H_n$ and edge weights as $p_{u,v}$.
%Using the estimated $p_{u,v}$, we assume a new complete graph with nodes same as that of $H_n$ and edge weights as $p_{u,v}$.
Let us now define a metric $p_u:=\sum_{v \in V} p_{u,v}$ for every node $u$ of $H_n$. Since $p_{u,v}$ denotes the probability that node $u$ is older than node $v$, $p_u$ would give a high score when a node $u$ becomes the oldest node. Our ranking is then sorted order of the $p_u$ values.

Instead of total order, a partial order can be found by a simple binning over $p_u$ values: fix the bin size $|C|$ and group $|C|$ nodes in the sorted $p_u$ values into a cluster, and the process repeats for other clusters. If $|C|=1$, the algorithm will yield a total order.

\xhdr{$p_{u,v}$-\texttt{threshold} algorithm} Here, each of the estimated $p_{u,v}$'s is compared against a threshold $\tau$. Only the node pairs that are strictly greater than this condition are put into the estimator output partial order. Note that if $\tau =0.5$, we get a total order in virtually all relevant cases.

\subsection{Semi-supervised solution}
\label{subsec:sup_soln}
Suppose we have partial true data available. Let it be ordered in partial order as $\sigma_{\text{orig}} = \{(u,v)\}$, in which for the pair $(u,v)$, $u$ is the older than $v$. Let $\sigma_{\text{train}} \subset \sigma_{\text{orig}}$ be the training set and the let the test set be $\sigma_{\text{test}}:= \sigma_{\text{orig}} \backslash \sigma_{\text{train}}$.  Let $|\sigma_{\text{train}}| = \alpha |\sigma_{\text{orig}}|$ for some $0<\alpha <1$. With the knowledge of $\sigma_{\text{train}}$, we modify the estimation of $p_{u,v}$ as follows. The set of removable nodes $\R_{H_s}$ at each instant $s$ is modified to $\R_{H_s} \backslash \NR_{H_s}$, where $\NR_{H_s}$ is the set of nodes that can not be included in the removable nodes as it would violate the partial order of $\sigma_{\text{train}}$. It is defined as follows:
\begin{equation}
\NR_{H_s}:= \{u: (u,v)\in \sigma_{\text{train}}, u,v \in V(H_s)  \}, \forall n \geq s \geq n_0.
\label{eq:NR}
\end{equation}
After estimating $p_{u,v}$ with the redefined $\R_{H_s}$, we employ \texttt{sort-by-$p_{u,v}$-sum} algorithm or {$p_{u,v}$-\texttt{threshold}} algorithms to find partial order.
An example of $\R_{H_s}$ construction is shown in Figure~\ref{fig:supervised_learning_example}.
\begin{figure}
%        \vspace*{-1 em}
\centering
%    \tikzstyle{every node}=[draw, circle, fill=black, inner sep=0pt, minimum width=3 pt]
%    \begin{tikzpicture}[scale=1]
%      \node (Y) at (0, 1) [label=above:$y$]{};
%      \node (U) at (0, 0) [label=below:$u$]{};
%      \node (V) at (1, 0) [label=below:$v$]{};
%      \node (W) at (2, 0) [label=below:$w$]{};
%      \node (X) at (3, 0) [label=below:$x$]{};
%      \node (Z) at (2, 1) [label=above:$z$]{};
%      \path[->] (U) edge[bend right=20] (V);
%      \path[->] (V) edge[bend right=20] (W);
%      \path[->] (W) edge[bend right=20] (X);
%      \path[->] (Y) edge[bend left=10] (X);
%    \end{tikzpicture}
\includegraphics[scale=0.5]{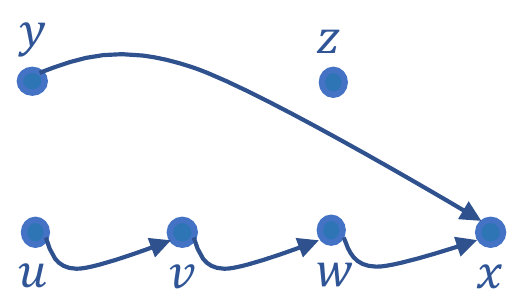}
%\vspace*{0 em}
\caption{Semi-supervised learning example DAG for $\sigma_{\text{train}} = \{(u, v), (v, w), (w, x), (y, w)$: $\NR_{H_s} = \{v, w, x\}$ and $\R_{H_s} = \{u, y, z\}$.}
\label{fig:supervised_learning_example}
%        \vspace*{0 em}
\end{figure}

Algorithm~\ref{alg:semi-supervised} summarizes our semi-supervised algorithm. The algorithm will become unsupervised when there is no $\sigma_{\text{train}}$ and step-5 is removed.

\section{Temporal Ordered Clustering for Duplication-Divergence Model}
\subsection{Duplication-divergence model (DD-model)}
\label{subsec:dd-model}
We consider Sol\'e et al.\ definition of the DD-model \cite{pastor2003evolving}. It proceeds as follows. Given an undirected, simple seed graph $G_{n_0}$
on $n_0$ nodes and target number of nodes $n$,
the graph $G_{k+1}$ with $k+1$ nodes\footnote{The subscript $k$ with $G_k$ can also be interpreted as time instant $k$} evolves from the $G_k$ as follows:
first, a new vertex $v$ is added to $G_k$. Then the following steps are carried out:
\begin{itemize}[leftmargin= 2ex,topsep= 0 ex]
    \item Duplication: Select a node $u$ from $G_k$ uniformly at random.
    The node $v$ then makes connections to $\cN(u)$, the neighbor set of $u$.
    \item Divergence: Each of the newly made connections from $v$ to $\cN(u)$ are deleted with probability $1-p$.
    Furthermore, for all the nodes in $G_k$ to which $v$ is not connected,
    create an edge from it to $v$ independently    with probability $\frac{r}{k}$.
\end{itemize}
The above process is repeated until the number of nodes in the graph is equal to $n$.
We denote the graph $G_n$ generated from the DD-model with parameters $p$ and $r$, starting from seed graph $G_{n_0}$, by $G_n \sim \ddmodel(n, p, r, G_{n_0})$.

The posterior probability $w(\delta(H_s,z_s),H_s)$, which is defined in \eqref{eq:expression_posterior_w} and used in Theorem~\ref{thm:Markov-chain-apprxn} and \texttt{high-prob-sampling}, can be calculated for the DD-model as follows. For a node $z_s \in  \mathcal{R}_{H_s}$, we say a node $u$ is its parent if $u$ can be selected from the graph $\delta(H_s, z_s)$ for the duplication step when $z_s$ is added into $\delta(H_s, z_s)$.
The probability of having the node $u$ as the parent of $z_s \in \mathcal{R}_{H_s}$ in the DD-model is
\begin{align}
w&(\delta(H_s,z_s), u, H_s)  \nn \\
& = \frac{1}{s-1} p^{|\cN(z_s) \cap \cN(u)|} (1-p)^{|\cN(u) \backslash \cN(z_s)|} \nn \\
& \quad \left(\frac{r}{s-1} \right)^{|\cN(z_s) \backslash \cN(u)|} \left(1-\frac{r}{s-1} \right)^{(s-1)-|\cN(z_s) \cup \cN(u)|}
\label{eq:w_defn}
\end{align}
The above expression can be inferred directly from the definition of the DD-model as follows.  
We first pick $u$ as a parent node of $z_s$ with probability $\frac{1}{s-1}$. 
Then to calculate the probabilities of edge addition events retrospectively, we observe that edges from  $z_s$ to the nodes in the $\cN(z_s) \cap \cN(u)$ stayed with probability $p$, but edges to $\cN(u) \backslash \cN(z_s)$ were dismissed with probability $1-p$. 
We also have to take into account the edges between $z_s$ and vertices outside of $\cN(u)$ -- each were chosen independently with probability $\frac{r}{s-1}$ and they are exactly the edges from $z_s$ to $\cN(z_s) \backslash \cN(u)$.

Now $w(\delta(H_s,z_s),H_s) = \sum_{u \in \mathcal{P}_{H_s}(z_s) } w(\delta(H_s,z_s), u, H_s) $, where $\mathcal{P}_{H_s}(z_s)$ represents possible parents of $z_s$ in $H_s$.

Since all permutations have positive probability in this version of the model, we have $\R_{H_s} = V(H_s)$ and $\Gamma(H_s) = s!$.
%\subsubsection{Pure duplication model}
%Pure-duplication model is a variation of the duplication-divergence model when $r=0$. Here a node $u$ can be a parent of another node $v$ only when all the neighbors of $v$ is contained in the neighborset of $u$.
%\begin{equation}
%w(\delta(H_t,v_t),H_t) = \frac{1}{t-1} p^{\deg(v_t)} \sum_{\mathclap{w \in \R_{H_t}(v_t)}} (1-p)^{\deg(w)-\deg(v_t)}
%\label{eq:w_defn}
%\end{equation}
%
%\begin{align}
%\R_{H_t} & = \{v \in V(H_t) - V(H_{t_0}): \exists w \in V(H_t) \text{ s.t. } \N(v) \subseteq \N(w)\}, \nn \\
%\R_{H_t}(v) &= \{w \in V(H_t) \text{ s.t. } \N(v) \subseteq \N(w)\}. \nn
%\end{align}
%Figure~\ref{fig:example_unif_sampling_10-02} shows an example.
%\begin{figure}[!htb]
%    \centering
%    \includegraphics[scale = 0.35]{./fig/10-2-06_node_order.png}
%    \caption{\texttt{DD-model}($n=10$, $G_{n_0} = K_2$, p = 0.6).}
%    \label{fig:example_unif_sampling_10-02}
%\end{figure}

\subsection{Greedy algorithms for clustering}
To form a comparison with algorithms proposed in Section~\ref{sec:approx_opt_soln}, we propose the following greedy unsupervised algorithms for the DD-model.

\xhdrLessSpace{\texttt{sort-by-degree}} The nodes are sorted by the degree and arranged into clusters $\{C_i\}_{i \geq 1}$. Cluster $C_1$ contains nodes with the largest degree.

\xhdrLessSpace{\texttt{peel-by-degree}} The nodes with the lowest degree are first collected and put in the highest cluster. Then they are removed from the graph, and the nodes with the lowest degree in the remaining graph are found and the process repeats.

\xhdrLessSpace{\texttt{sort-by-neighborhood}} This algorithm will output a partial order with all ordered pairs $(u \prec v)$ such that $\N(u)$ contains $\N(v)$. This condition holds when $r=0$. When $r>0$, we consider $|\N(v) \backslash N(u)| \leq r$ as $r$ is the average number of {\em extra} connections a node makes apart from duplication process. In most real-world data, we estimate $r$ as smaller than $1$, and hence the original check is sufficient.

\xhdrLessSpace{\texttt{peel-by-neighborhood}} Here, we find the set
$\{u: \nexists\, v |\N(v) \backslash \N(u)| \leq r \}$
(as mentioned before, it is sufficient to check $\N(v) \subset \N(u)$ in many practical cases) and mark it as the youngest cluster. These nodes are removed from the graph, and the process is repeated until it hits $G_{n_0}$. This algorithm makes use of the DAG of the neighborhood relationship and includes isolated nodes into the bins.

\subsection{Comparison with other graph models}
\label{subsec:comparison-with-other-models}
The node arrival order recovery problem in the DD-model is different from that in other graph models like Erd\H{o}s-Renyi graphs and preferential attachment graphs.

First, for a fixed graph $G_n$ on $n$ vertices, let us consider a set of graphs $\text{Adm}(G_n) = \{\sigma(G_n)\colon \sigma \in \Gamma(G_n)\}$. It is obvious that for the Erd\H{o}s-Renyi model, any graph in $\text{Adm}(G_n)$ is generated equally likely with a given seed graph $G_{n_0}$. Such property was also proved for the preferential attachment model in \cite{luczak2016asymmetry}. However, this does not hold for DD-model graphs as shown in the following example.
\begin{figure}[!htb]
    %\vspace*{-1 em}
    \centering
%        \tikzstyle{every node}=[draw, circle, fill=black, inner sep=0pt, minimum width=3 pt]
%        \begin{tikzpicture}[yscale=1.5, xscale=1.2]
%          \node (V1) at (-1, 0) [label=below:$1$]{};
%          \node (V2) at (0, 2) [label=above:$2$]{};
%          \node (V3) at (1, 0) [label=below:$3$]{};
%          \node (V4) at (3, 2) [label=above:$4$]{};
%          \node (V5) at (3, 0) [label=below:$5$]{};
%          \node[draw=white,fill=white] (PH) at (0, -1) {};
%          \draw (V2) -- (V3) -- (V1) -- (V2) -- (V4);
%        \end{tikzpicture}
\includegraphics[scale=0.65]{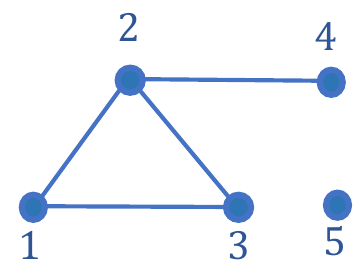}
        %\vspace*{0 em}
        \caption{Example of asymmetric graph}
        \label{fig:equiprobability}
        %\vspace*{0 em}
    \end{figure}

For the graph $G_n^{(1)}$ presented in Figure~\ref{fig:equiprobability}, let $G_{n_0}$ consists of vertices $1$, $2$, $3$, and let the parameters of the DD-model be $p = 0.2$ and $r = 0$. 
The $\P[\cG_n = G_n^{(1}]$ can be calculated iteratively using \eqref{eq:w_defn} as $0.068$. Now, consider the permutation
\[
\sigma = \begin{pmatrix} 1 & 2 & 3 & 4 & 5\\ 1 & 2 & 3 & 5 & 4 \end{pmatrix}.\]
Then $\sigma \in \Gamma(G_n^{1})$. Let $G_n^{(2)}= \sigma(G_n^{(1)})$. The $\P[\cG_n = G_n^{(2}]$ is $0.051$, and conditioned on the same structure probabilities of $G_n^{(1)}$ and $G_n^{(2)}$ are $0.5744$ and $0.4256$ respectively.
    \begin{figure}[!htb]
        %\vspace*{-1 em}
        \centering
        \includegraphics[scale=1.0]{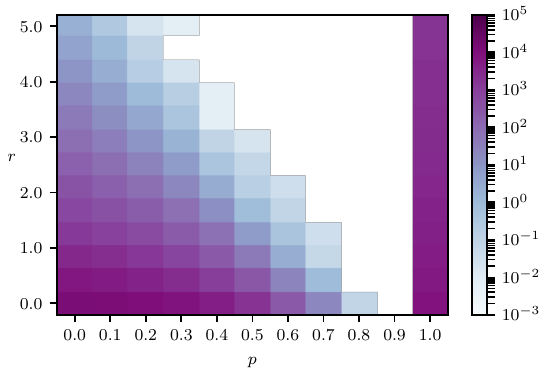}
        %\vspace*{-1 em}
        \caption{$\E \log |\Aut(G_n)|$, $G_n \sim \ddmodel(2000, p, r, K_{20})$, where $|\Aut(G_n)|$ is the number of automorphisms in graph $G_n$.}
        \label{fig:symmetry}
    %\vspace*{0 em}
%    \caption{Auxiliary figures for comparison duplication-divergence with other models.}
    %\vspace*{-2 em}
\end{figure}

% For the graph $G_t$ presented above, $G_0$ consisting of vertices $1$, $2$, $3$, and the parameters $p = 0.2$ and $r = 0$, the total probability of generating a structure identical to the presented above is equal to:
% \begin{align*}
%     \frac{2}{3} \cdot 0.2 \cdot 0.8 \left(\frac{1}{4} \cdot 0.8 + \frac{1}{2} \cdot 0.8^2 + \frac{1}{4} \cdot 0.8^3\right) + \frac{3}{3} \cdot 0.8^2 \cdot \frac{1}{2} \cdot 0.2 \cdot 0.8 = 0.1203.
% \end{align*}

% Moreover, the probability of adding vertices in a sequence ($4$, $5$) is equal to:
% \begin{align*}
%     \frac{1}{0.224} \cdot \frac{2}{3} \cdot 0.2 \cdot 0.8 \left(\frac{1}{4} \cdot 0.8 + \frac{1}{2} \cdot 0.8^2 + \frac{1}{4} \cdot 0.8^3\right) = 0.7714
% \end{align*}
% but the probability of adding vertices in a sequence ($5$, $4$) is equal to:
% \begin{align*}
%     \frac{1}{0.224} \cdot \frac{3}{3} \cdot 0.8^2 \frac{1}{2} \cdot 0.2 \cdot 0.8 = 0.2285.
% \end{align*}
Second, it is well known that both the Erd\H{o}s-Renyi graphs and preferential attachment graphs are {\em asymmetric}\footnote{An automorphism or symmetry of a graph $G$ is an isomorphism from a graph $G$ to itself.
We say that $G$ is symmetric if it has at least one nontrivial symmetry and that $G$ is asymmetric if the only symmetry of $G$ is the identity permutation.} with high probability \cite{kim2002asymmetry,luczak2016asymmetry}.
On the other hand, the graphs generated from the DD-model for a certain range of parameters show a significant amount of symmetry, as shown in Fig. \ref{fig:symmetry}. This is in accordance with many real-world networks (see Table \ref{tab:real-world_est_parameters} for examples).

Last, the behavior of the degree at time $t$ of the node arrived at an earlier time $s$ (denoted by $\deg_t(s)$) is different for all three models. For Erd\H{o}s-Renyi graph with edge probability $p$, it is known that $\E [\deg_t(s)] = p (t - 1)$. For the preferential attachment graphs, $\E[ \deg_t(s)] = \Theta\left(\sqrt{t/s}\right)$ (\cite{hofstad2016random}, Theorem 8.2).
However, for the DD-model, $\E[ \deg_t(s)] = \Theta\left(\left(t/s\right)^p s^{2 p - 1}\right)$ for any $t \ge s$~\cite{psexpected}. Note that when $s=O(1)$ -- the case of very old nodes -- the average degree is of order $t^p$. For $s=t$, we have $\E[\deg_t(t)]=O(t^{2p-1})$ which is growing only for $p>1/2$. For example, when $p=1$ degrees of all the nodes on average are of order $O(t)$. Thus oldest nodes in the graph need not have large average degrees as the graph evolves, and algorithms based on such a heuristic are not applicable for the DD-model.
Moreover, Frieze et al. \cite{frieze2020} has shown that $\deg_t(s)$ is concentrated around the mean for $s = O(1)$ in the sense that for any $A > 1$ we observe polynomial tail:
\begin{align*}
    \Pr&[\text{deg}_t(s) < C\, \E[\text{deg}_t(s) \log^{-k}{t}] \\
    & = \Pr[\text{deg}_t(s) > C\, \E[\text{deg}_t(s) \log^k{t}] = O(t^{-A}),
\end{align*}
for certain fixed constant $k$ and a constant $C$ dependent on $A$.
However, this is not the case for the last vertices, since they are copied from already existing nodes in the network, which would explain the ineffectiveness of greedy degree-based heuristics for $p \le 1/2$.

\section{Experiments}
In this section, we evaluate our methods on synthetic and real-world data sets. We made publicly available all the code and data of this project at \url{https://github.com/krzysztof-turowski/duplication-divergence}.

We present the following results in the coming sections.
\begin{itemize}[leftmargin= 2ex,topsep=1 ex]
    \item Synthetic networks:
    \begin{itemize}
      \item How well the LP-partial-order performs in comparison with the LP-clusters? (Figures~\ref{fig:temporal_curve_LP_clustering_LP_partial_order} and \ref{fig:time_comparison_LP_clustering_LP_partial_order})
      \item Fixing the LP-partial-order, how is the convergence of $p_{u,v}$'s that are estimated via sequential importance sampling schemes \texttt{local-unif-sampling} and \texttt{high-prob-sampling} as to the exact $p_{u,v}$, in terms of resulting precision? (Figure~\ref{fig:synthetic_graph_exact})
      \item Fixing LP-partial-order for the LP formulation and \texttt{local-unif-sampling} for the importance sampling strategy, we study the performance of unsupervised algorithms in comparison with greedy strategies specific to the DD-model. (Figure~\ref{fig:synthetic_graph_algorithms})
      \item For the semi-supervised algorithms, we show results (precision and density) for various parameter configurations and study their influence on the performance. (Tables~\ref{tab:synthetic-supervised} and \ref{tab:syn_data_change})
    \end{itemize}
    \item Real-world networks: For the semi-supervised algorithms, how the precision improves with a small change in the training size, and how does the results compare against greedy algorithms of the DD-model? (Figure~\ref{fig:real-world-nw} and Table~\ref{tab:real-world-tab})
\end{itemize}

\subsubsection*{Maximum likelihood estimation}
For deriving total order, a natural solution will be the maximum likelihood estimator (MLE).
\[ \argmax_{\sigma \in S_n} \P[\cG_n = \bpi^{-1}(H_n)| \bpi^{-1} = \sigma] \]
But we do not consider MLE explicitly here because it is known that many networks exhibit large number of symmetries (see Table~\ref{tab:real-world_est_parameters} for some examples), and thus there will be large number of total orders that achieve the MLE criterion with low value of precision. In fact, our optimal formulation in Section~\ref{sec:formln} already captures the MLE solutions and outputs them if they have high precision. Moreover for general graph models, the MLE computation would require checking all $\sigma \in S_n$ which incurs $\Theta(n!)$ computational complexity.

\subsection{Synthetic networks}
\label{subsec:synthetic_data_numerical}

In the following results on synthetic networks, $\sigma_{\text{tries}}$ denote the number of Markov chain sample paths (for sequential importance sampling) used for estimating $p_{u,v}$ for all $u,v \in V(H_n)$. All the studies are performed on multiple graph realizations from the DD-model with specified parameters, and the results are averaged over them.
When we make a comparison based on LP formulation, we plot precision ($\theta$) vs minimum density $\varepsilon$ ($\delta \geq \varepsilon$) in accordance with the formulations in Sections~\ref{subsec:optzn_cluster} and \ref{subsec:optzn_partial_orders}.

\begin{figure}[!htb]
    \centering
    \includegraphics[scale=1.0]{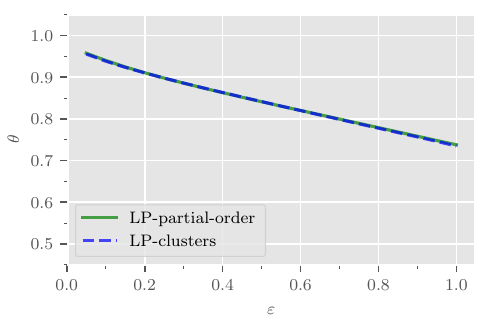}
    \caption{Comparison between LP-clusters and LP-partial-order formulation: $G_n \sim \ddmodel(n=30,p=0.6,r=1.0, G_{n_0}=K_{10})$ and $\sigma_{\text{tries}} = 100,000$. Results are averaged over $100$ graph generations. Sampling method: \texttt{local-unif-sampling}.}
    \label{fig:temporal_curve_LP_clustering_LP_partial_order}
    % \vspace{-2 em}
\end{figure}
\begin{figure}[!htb]
    \centering
    \includegraphics[scale=1.0]{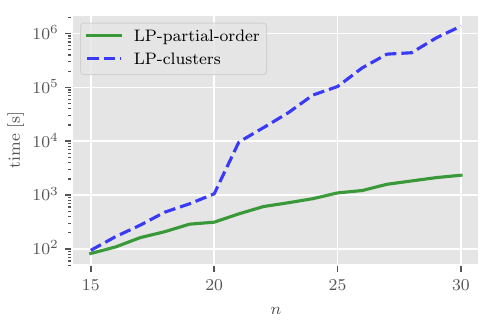}
    \caption{Time plot for LP-partial-order vs. LP-clusters. All the experiments were performed on 48-CPU cluster, with Intel(R) Xeon(R) CPU E7-8857 v2 @ 3.00GHz and 256GB RAM.}
    \label{fig:time_comparison_LP_clustering_LP_partial_order}
    % \vspace{-1 em}
\end{figure}

In Figure~\ref{fig:temporal_curve_LP_clustering_LP_partial_order}, we compare the performance of the linear programming approximations LP-cluster (Section~\ref{subsec:optzn_cluster}) and LP-partial-order (Section~\ref{subsec:optzn_partial_orders}). Since clustering output from the LP-cluster scheme induces a partial order, we use the same measures of precision and density that are defined for partial order for comparing performances of LP-cluster and LP-partial-order schemes.
Our experiments confirm that for the same graph, with the same set of $\{p_{u,v}, \forall u, v \in V(H_n)\}$, the performance of them are nearly identical.
% The reason is that the actual solutions found by both the LP relaxations are indeed very close to the solutions of both original integer programs, which have the identical optimum values of the objective function -- as the true underlying partial order is derived from the true binning.
However, Figure~\ref{fig:time_comparison_LP_clustering_LP_partial_order} shows that the difference between the running time of both the formulations is huge -- the LP-clusters which finds clusters becomes barely feasible, whereas LP-partial-order which outputs partial order runs in a reasonable time.

\begin{figure*}[htb!]
    %\vspace*{-1 em}
    \centering
    \begin{subfigure}{0.48\textwidth}
        \centering
        \includegraphics[scale=1.0]{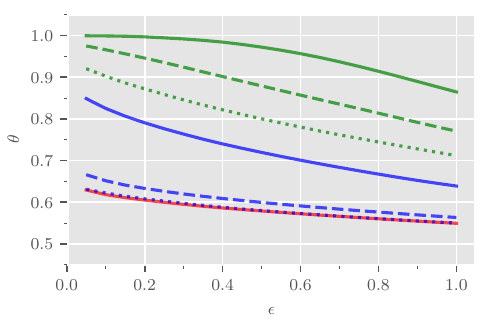}
        %\vspace{-1 em}
    \end{subfigure}
    \begin{subfigure}{0.48\textwidth}
        \centering
        \includegraphics[scale=1.0]{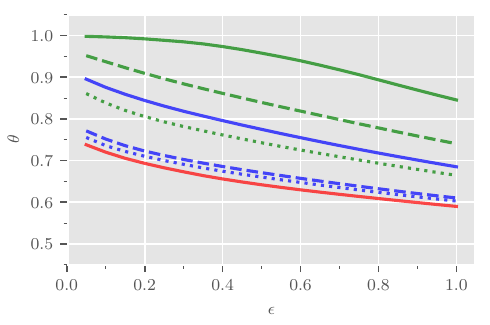}
        %\vspace{-1 em}
    \end{subfigure}
    \begin{subfigure}{1\textwidth}
        \centering
        \includegraphics[scale=1.0]{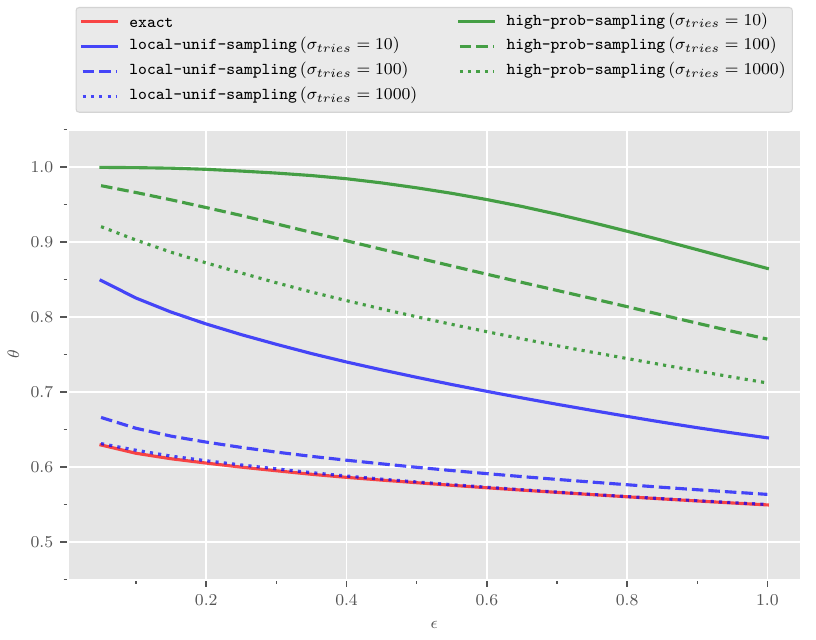}
    \end{subfigure}
    %\vspace*{0 em}
    \caption{Results on synthetic networks with exact curve: $G_n \sim \ddmodel(13, p, 1.0, G_{n_0})$ for $p = 0.3$ (left) and $0.6$ (right), averaged over $100$ graphs. $G_{n_0}$ is generated from Erd\H{o}s-Renyi graph with $n_0 = 4$ and $p_0=0.6$.}
    \label{fig:synthetic_graph_exact}
    %\vspace*{-1 em}
\end{figure*}
\begin{figure*}[htb!]
    %\vspace*{0 em}
    \centering
    \begin{subfigure}{0.48\textwidth}
        \centering
        \includegraphics[scale=1.0]{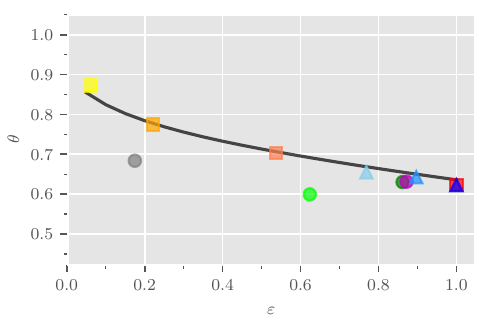}
        %\vspace{-1 em}
    \end{subfigure}
    \begin{subfigure}{0.48\textwidth}
        \centering
        \includegraphics[scale=1.0]{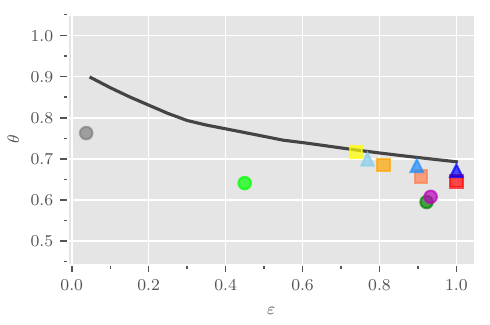}
        %\vspace{-1 em}
    \end{subfigure}
    \begin{subfigure}{1\textwidth}
        \centering
        \includegraphics[scale=1.0]{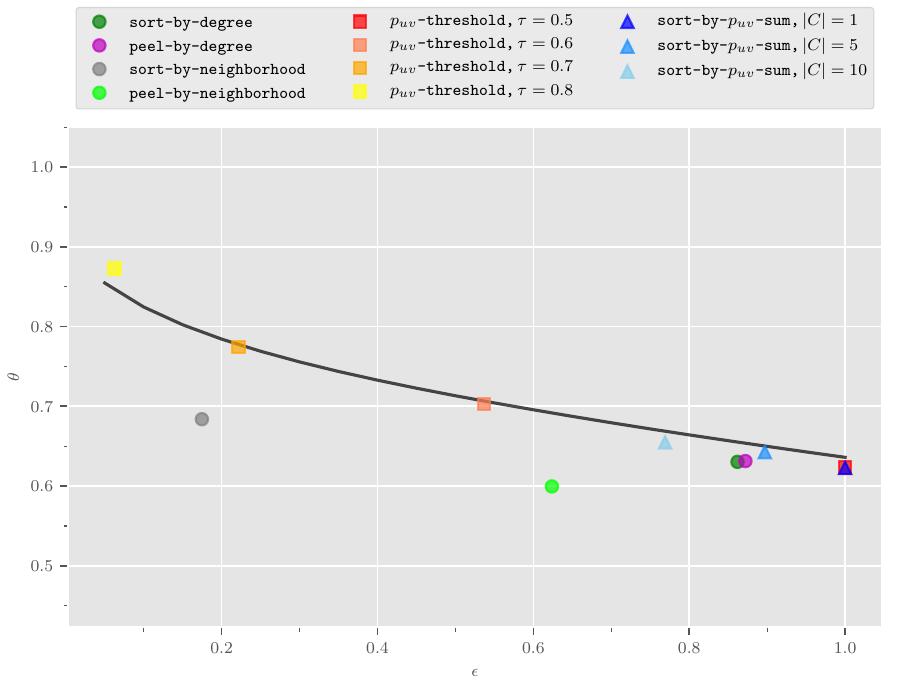}
    \end{subfigure}
    %\vspace*{0 em}
    \caption{Results on synthetic networks with greedy and  {\em unsupervised learning} $p_{u,v}$-based algorithms: $G_n \sim \ddmodel(50, p, 1.0, G_{n_0})$ for $p = 0.3$ (left) and $0.6$ (right), averaged over $100$ graphs. $p_{u,v}$-based algorithms use $\sigma_{\text{tries}} = 100,\!000$. $G_{n_0}$ is generated from Erd\H{o}s-Renyi model with $n_0 = 10$ and $p_0=0.6$. The theoretical curve is estimated via \texttt{local-unif-sampling}.}
    \label{fig:synthetic_graph_algorithms}
    \vspace*{-1 em}
\end{figure*}

% In Figures~\ref{fig:synthetic_graph_exact} and \ref{fig:synthetic_graph_algorithms}, we provide the linear programming (LP) {\em optimal curve} or an approximation to it by one of the previously presented sampling methods.
% The LP optimal curve is given by the optimal precision values computed from the relaxed LP formulation and exact $p_{u,v}$ values for various values of $\varepsilon$, which corresponds to the density lower bound of $\epsilon$.

Figure~\ref{fig:synthetic_graph_exact} examines the precision of LP-partial-order obtained with approximated $\{p_{u,v}, \forall u, v \in V(H_n)\}$ via sequential importance strategies (\texttt{local-unif-sampling} and \texttt{high-prob-sampling}) and that obtained with the $\{p_{u,v}, \forall u, v \in V(H_n)\}$ that is calculated exactly by considering all the possible $n!$ orderings.
We consider a small size ($n=13$) example here since it becomes infeasible to compute the exact curve for larger values of $n$.
We observe that the convergence of the estimated curve is highly dependent on the method of estimation: \texttt{local-unif-sampling} method requires only $100$ samples, but \texttt{high-prob-sampling} is still visibly far away from LP optimal curve even for $1000$ samples.
Thus, along with the computational reasons stated in Section~\ref{sec:approx_opt_soln}, we use \texttt{local-unif-sampling} in the subsequent experiments.

In Figure~\ref{fig:synthetic_graph_algorithms}, we compare results of the unsupervised algorithms with the estimated optimal curve via \texttt{local-unif-sampling}.
It turns out that greedy algorithms (\texttt{sort-by-degree}, \texttt{sort-by-neighborhood}, \texttt{peel-by-degree} and \texttt{peel-by-neighborhood}) perform reasonably well for small $p$, but their performance deteriorates for higher values of $p$.
On the other hand, $p_{u,v}$-based algorithms (\texttt{sort-by-$p_{u,v}$-sum} and \texttt{$p_{u,v}$-threshold}) offer consistent, close to the theoretical bound, behavior different values of $p$ (figure shows only two $p$s due to space limitations). Moreover, the bin size $|C|$ in \texttt{sort-by-$p_{u,v}$-sum} and threshold $\tau$ in \texttt{$p_{u,v}$-threshold} algorithm offer a trade-off between higher precision and higher density.
% It is worth noting that the free parameters in both $p_{u,v}$-based algorithms serve as a trade-off:
The larger the bin size or the higher the threshold, we observe a decrease in density, but increase in precision as we stay close to the theoretical curve.
% We also observe that for grater values of $p$ this trade-off diminishes visibly.
% \JS{last line is a bit confusing. Try to remove it. Change language as in ISIT paper}

Table \ref{tab:synthetic-supervised} contains the results of semi-supervised learning extensions of the $p_{u,v}$-based algorithms.
A small increase in the percentage of the training set $\alpha$ yields a large increases in precision for all sets of parameters.
Moreover, for larger bin size in \texttt{sort-by-$p_{u,v}$-sum} algorithm we observe mainly only an increase in precision with $\alpha$, but for \texttt{$p_{u,v}$-threshold} algorithm both $\delta$ and $\theta$ grow visibly with $\alpha$, especially for large $\tau$.
In turn, when we fix $\alpha$ and increase the bin size in \texttt{sort-by-$p_{u,v}$-sum} algorithm, precision remains almost same, but density decreases significantly. And if we do the analogous procedure for \texttt{$p_{u,v}$-threshold} algorithm (fix $\alpha$ and increase threshold), then precision grows, but in expense of a visible fall of density.
All the above conclusions are summarized in Table~\ref{tab:syn_data_change}.
\begin{table}[!htb]
    % \vspace*{-1 em}
    \centering
    \begin{tabular}{@{}lccccc@{}}\toprule
        &  & \multicolumn{2}{c}{$p = 0.3$} & \multicolumn{2}{c}{$p = 0.6$}\\
        \cmidrule{3-4} \cmidrule{5-6}
        Algorithm & $\alpha$ & $\delta$ & $\theta$ & $\delta$ & $\theta$ \\ \midrule
        \rowcolor{gray!25}    \texttt{sort-by-$p_{u,v}$-sum}, $|C| = 1$ & $0.001$ & $1.0$ & $0.598$ & $1.0$ & $0.613$ \\
        \rowcolor{gray!25}    \texttt{sort-by-$p_{u,v}$-sum}, $|C| = 1$ & $0.01$ & $1.0$ & $0.643$ & $1.0$ & $0.650$ \\
        \rowcolor{gray!25}    \texttt{sort-by-$p_{u,v}$-sum}, $|C| = 1$ & $0.1$ & $1.0$ & $0.836$ & $1.0$ & $0.832$ \\
%        \texttt{sort-by-$p_{u,v}$-sum}, $|C| = 5$ & $0.001$ & $0.897$ & $0.603$ & $0.897$ & $0.624$ \\
%        \texttt{sort-by-$p_{u,v}$-sum}, $|C| = 5$ & $0.01$ &$0.897$ & $0.656$ & $0.897$ & $0.651$ \\
%        \texttt{sort-by-$p_{u,v}$-sum}, $|C| = 5$ & $0.1$ & $0.889$ & $0.857$ & $0.889$ & $0.864$ \\
        \texttt{sort-by-$p_{u,v}$-sum}, $|C| = 10$ & $0.001$ & $0.769$ & $0.605$ & $0.769$ & $0.626$ \\
        \texttt{sort-by-$p_{u,v}$-sum}, $|C| = 10$ & $0.01$ & $0.768$ & $0.661$ & $0.767$ & $0.660$ \\
        \texttt{sort-by-$p_{u,v}$-sum}, $|C| = 10$ & $0.1$ &$0.758$ & $0.864$ & $0.759$ & $0.859$ \\
        \rowcolor{gray!25} \texttt{$p_{u,v}$-threshold}, $\tau = 0.5$ & $0.001$ &$1.0$ & $0.604$ & $1.0$ & $0.617$ \\
        \rowcolor{gray!25} \texttt{$p_{u,v}$-threshold}, $\tau = 0.5$ & $0.01$ &$1.0$ & $0.637$ &$1.0$ & $0.649$ \\
        \rowcolor{gray!25}\texttt{$p_{u,v}$-threshold}, $\tau = 0.5$ & $0.1$ &$1.0$ & $0.829$ &$1.0$ & $0.823$ \\
%        \rowcolor{gray!25}        \texttt{$p_{u,v}$-threshold}, $\tau = 0.7$ & $0.001$ & $0.120$ & $0.754$ & $0.219$ & $0.755$ \\
%        \rowcolor{gray!25}        \texttt{$p_{u,v}$-threshold}, $\tau = 0.7$ & $0.01$ & $0.287$ & $0.791$ & $0.382$ & $0.770$ \\
%        \rowcolor{gray!25}        \texttt{$p_{u,v}$-threshold}, $\tau = 0.7$ & $0.1$ & $0.778$ & $0.899$ & $0.800$ & $0.888$ \\
        \texttt{$p_{u,v}$-threshold}, $\tau = 0.9$ & $0.001$ & $0.010$ & $0.906$ & $0.028$ & $0.871$ \\
    \texttt{$p_{u,v}$-threshold}, $\tau = 0.9$ & $0.01$ & $0.020$ & $0.951$ & $0.090$ & $0.907$ \\
\texttt{$p_{u,v}$-threshold}, $\tau = 0.9$ & $0.1$ & $0.521$ & $0.966$ & $0.559$ & $0.960$ \\
        \bottomrule
    \end{tabular}
    %\vspace*{0.5em}
    \caption{Results on synthetic networks with  {\em semi-supervised learning} $p_{u,v}$-based algorithms: $G_n \sim \ddmodel(50, p, 1.0, G_{n_0})$, averaged over $100$ graphs. $p_{u,v}$-based algorithms use $\sigma_{\text{tries}} = 100,000$. $G_{n_0}$ is Erd\H{o}s-Renyi graph with $n_0 = 10$ and $p_0=0.6$.}
    \label{tab:synthetic-supervised}
    %\vspace*{0 em}
\end{table}
\begin{table}[!htb]
    % \vspace*{-0.5 em}
    \centering
    \begin{tabular}{lccccc} \toprule
        %        &  \phantom{abc} & \multicolumn{4}{c}{Parameters} \\
        %\cmidrule{2-5}
        Algorithm & Fixed & Free & Free & Free \\
        \midrule
        \texttt{sort-by-$p_{u,v}$-sum} & $\alpha$  & $|C| \nearrow$    & $\delta \searrow$ & $\theta \approx$  \\
        \texttt{sort-by-$p_{u,v}$-sum} & $|C|$    & $\alpha \nearrow$ &  $\delta \approx$  & $\theta \nearrow$ \\
        \texttt{$p_{u,v}$-threshold}   & $\alpha$ & $\tau \nearrow$   & $\delta \searrow$ & $\theta \nearrow$ \\
        \texttt{$p_{u,v}$-threshold}   & $\tau$ & $\alpha \nearrow$ & $\delta \nearrow$ & $\theta \nearrow$ \\
        \bottomrule
    \end{tabular}
    %\vspace*{0.5em}
    \caption{Conclusions from synthetic data: how the metrics behave by fixing one of the parameters and keeping other free. The symbol $\approx$ indicates the changes are not significant.}
    \label{tab:syn_data_change}
        \vspace{-2em}
\end{table}
% Big picture, just in case
% \begin{figure*}[htb!]
    % \centering
    % \includegraphics[scale=0.8]{./newfig/synthetic-13-4-PS-0600-100.pdf}
    % \caption{Results on synthetic networks with exact curve. $G_n \sim \ddmodel(13, 0.6, 1, G_{n_0})$, averaged over $100$ tries. $G_{n_0}$ is generated from Erd\H{o}s-Renyi graph with $n_0 = 4$ and $p_0=0.6$.}
    % \label{fig:synthetic_graph_exact}
% \end{figure*}
% \vspace*{-1 em}
\subsection{Real-world networks}
    \begin{figure}[htb!]
    \centering
    \begin{subfigure}{0.48\textwidth}
        \centering
        \includegraphics[scale=0.9]{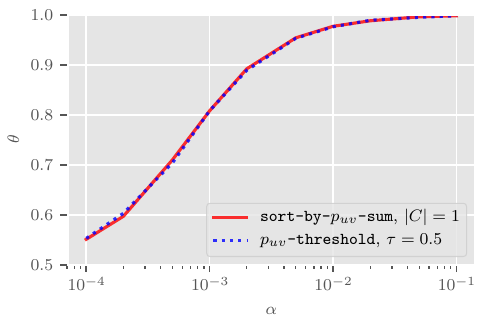}
        %\vspace*{0 em}
        \caption{ArXiv network}
        \label{fig:arxiv}
    \end{subfigure}
    \begin{subfigure}{0.48\textwidth}
        \centering
        \includegraphics[scale=0.9]{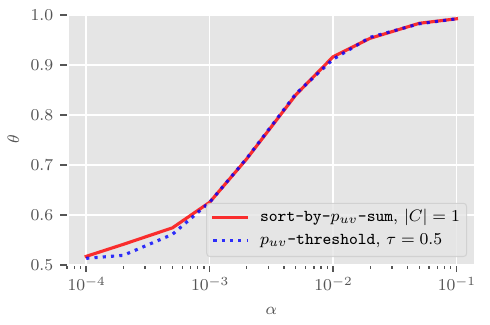}
        %\vspace*{0 em}
        \caption{Simple English Wikipedia network}
        \label{fig:simple_english}
    \end{subfigure}
    \begin{subfigure}{0.48\textwidth}
        \centering
        \includegraphics[scale=0.9]{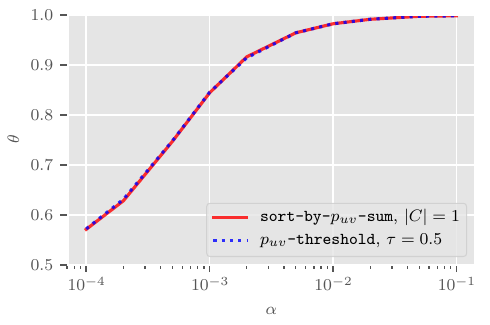}
        %\vspace*{0 em}
        \caption{CollegeMsg network}
        \label{fig:collegemsg}
    \end{subfigure}
    %\vspace*{0 em}
    \caption{Real-world networks: results of semi-supervised learning}
    \label{fig:real-world-nw}
    \vspace{-1 em}
\end{figure}

% \textcolor{blue}{The numerical comparisons with the real-world data in this section are with respect to the greedy heuristic algorithms since there are no other existing techniques to estimate temporal ordered clustering, according to our knowledge.}
We consider the following three real-world networks which have the ground truth of node and edge age arrival order available. The directed networks are treated as undirected in our studies. All the datasets are taken from SNAP repository~\cite{snapnets}.

\begin{itemize}
\item{\em The ArXiv network}: It is a directed network with $7,\!464$ nodes and $116,\!268$ edges. Here the nodes are the publications in arXiv online repository of theoretical high energy physics, and the edges are formed when a publication cite another. In this network, many nodes share the same arrival time and date, and hence the true arrival order of nodes is available only in bins of count $1,\!457$.

\item{\em The Simple English Wikipedia dynamic network}: A directed network with $10,\!000$ nodes and $169,\!894$ edges. Nodes represent articles and an edge indicates that a hyperlink was added. It shows the evolution of hyperlinks between articles of the Simple English Wikipedia.

\item{\em CollegeMsg network}: In this dataset of private message sent on an online social platform at University of California, Irvine, nodes represent users and an edge from $u$ to $v$ indicates user $u$ sent a private message to user $v$ at time $t$. Number of nodes is $1,\!899$ and number of edges is $59,\!835$.
\end{itemize}

Table~\ref{tab:real-world_est_parameters} shows estimated parameters of the duplication-divergence model for the above networks using the fitting technique in \cite{sreedharan2020revisiting}.
\begin{table}[htb!]
%\vspace{-0.5  em}
    \centering
    \small{
        \begin{tabular}[htb!]{lcccc}
            \toprule
            Network ($G_{\text{obs}}$) & $\log |\Aut(G_{\text{obs}})|$ & $\widehat{p}$ & $\widehat{r}$\\
            \midrule
            ArXiv & $12.59$ & $0.72$ & $1.0$\\
            Wikipedia & $1018.94$ & $0.66$ & $0.5$  \\
            CollegeMsg & $231.54$ & $0.65$ & $0.45$  \\
            \bottomrule
    \end{tabular}}
    %\vspace{0.5 em}
    \caption{Parameters of the duplication-divergence model estimated for the real-world networks considered in this paper.}
    \label{tab:real-world_est_parameters}
    %\vspace{0 em}
\end{table}

Figure~\ref{fig:real-world-nw} show the result of semi-supervised learning. Here $\alpha$ represents the proportion of all pairs that is considered as training set, i.e., size of the training set is $\alpha \binom{n}{2}$. We randomly pick $\alpha \binom{n}{2}$ pairs and the results presented are average over $100$ different such random sets. We observe that a small increase in $\alpha$ leads to a huge change in the precision. This also happens in synthetic data and is caused by the large structural dependency within networks, unlike in classical machine learning where data is often assumed to be independent. This helps us to get a near-perfect clustering (precision close to $1$) with only $1\%$ of the labeled nodes.

Finally, the semi-supervised approach helps to obtain a significant improvement over greedy algorithms.
As it is shown in Table \ref{tab:real-world-tab}, greedy algorithms found orderings with precision ranging from $0.47$ to $0.63$ for significant values of the density (it's easy to achieve a precision of $0.78$ like CollegeMsg data set when the density of pairs outputted is as low $0.01$) -- which is not much better than random guess.
%\begin{table*}[htb!]
%%\vspace{-1.25 em}
%\centering
%\begin{tabular}[b]{lccccccccccccccc}
%    \toprule
%    & \phantom{abc} & \multicolumn{2}{c}{ArXiv} & \phantom{abc} & \multicolumn{2}{c}{Simple English Wikipedia} & \phantom{abc} & \multicolumn{2}{c}{CollegeMsg} \\
%    \cmidrule{3-4} \cmidrule{6-7} \cmidrule{9-10}
%    Greedy algorithm & & $\delta$ &$\theta$ & & $\delta$ & $\theta$ & & $\delta$ & $\theta$ \\
%    \midrule
%    \texttt{sort-by-degree} & & $0.983$ & $0.477$ & & $0.962$ & $0.594$ & & $0.921$ & $0.633$ \\
%    \texttt{peel-by-degree} & & $0.989$ & $0.469$ & & $0.963$ & $0.593$ & & $0.922$ & $0.637$ \\
%    \texttt{sort-by-neighborhood} & & $0.0001$ & $0.514$ & & $0.036$ & $0.603$ & & $0.017$ & $0.788$ \\
%    \texttt{peel-by-neighborhood} & & $0.131$ & $0.509$ & & $0.807$ & $0.593$ & & $0.754$ & $0.615$ \\
%    \bottomrule
%\end{tabular}
%\caption{Real-world networks: results for greedy algorithms}
%\label{tab:real-world-tab}
%\end{table*}
\begin{table}[htb!]
    %\vspace{-0.5 em}
    \centering
    \begin{tabular}[b]{l@{\hspace{1\tabcolsep}}c@{\hspace{0.8\tabcolsep}}cc@{\hspace{0.8\tabcolsep}}cc@{\hspace{0.8\tabcolsep}}c}
        \toprule
        & \multicolumn{2}{c}{ArXiv} & \multicolumn{2}{c}{Wikipedia} & \multicolumn{2}{c}{CollegeMsg} \\
        \cmidrule{2-3} \cmidrule{4-5} \cmidrule{6-7}
        Greedy algorithm & $\delta$ &$\theta$ & $\delta$ & $\theta$ & $\delta$ & $\theta$ \\
        \midrule
        \texttt{sort-by-degree} & $0.98$ & $0.47$ & $0.96$ & $0.59$ & $0.92$ & $0.63$ \\
        \texttt{peel-by-degree} &  $0.98$ & $0.46$ & $0.96$ & $0.59$ & $0.92$ & $0.63$ \\
        \texttt{sort-by-neighborhood} &  $0.0001$ & $0.51$ & $0.03$ & $0.60$ & $0.01$ & $0.78$ \\
        \texttt{peel-by-neighborhood} & $0.13$ & $0.50$ & $0.80$ & $0.593$ & $0.75$ & $0.61$ \\
        \bottomrule
    \end{tabular}
    \caption{Real-world networks: results of greedy algorithms.}
    \label{tab:real-world-tab}
\vspace*{-2 em}
\end{table}

\section{Discussion and future work}

In this article we presented a framework for clustering of nodes in dynamic networks based on latent temporal information.
% and for using the semi-supervised learning algorithms in those types of problems.
We provided a way to find an upper bound on the optimum clustering quality, and proposed several algorithms that perform well and capable of including some external information about the precedence of vertices in their arrival to the network.
% We provided a way to find an upper limit of the informational content of the evolution of a dynamic network and certain algorithm which perform well and are capable of including some external information about the precedence of vertices.

Further work in our proposed framework can go in several directions. For example, one can explore various ways to speed up the algorithms presented in this work. This can be accomplished by finding a good importance sampling distribution, which will lead to a faster convergence of estimates of $p_{u, v}$ (probability that node $u$ is arrived earlier than node $v$), for all the nodes $u$ and $v$ in the network, to the true values.
From a theoretical perspective, there remains an interesting question of finding bounds on the convergence speed of estimates of ${p}_{u, v}$ with various importance sampling distributions.
One can also look into clever bookkeeping techniques which will result in reducing the computation time of a single path in our sequential importance sampling algorithm.
Another direction is the application of the proposed framework and solution to other types of random network models that not only involve only the addition of vertices and edges, but also deletion of them.

\bibliographystyle{IEEEtran}
\bibliography{nodeorder_DD.bib}

% Generated by IEEEtran.bst, version: 1.14 (2015/08/26)
\begin{thebibliography}{10}
\providecommand{\url}[1]{#1}
\csname url@samestyle\endcsname
\providecommand{\newblock}{\relax}
\providecommand{\bibinfo}[2]{#2}
\providecommand{\BIBentrySTDinterwordspacing}{\spaceskip=0pt\relax}
\providecommand{\BIBentryALTinterwordstretchfactor}{4}
\providecommand{\BIBentryALTinterwordspacing}{\spaceskip=\fontdimen2\font plus
\BIBentryALTinterwordstretchfactor\fontdimen3\font minus
  \fontdimen4\font\relax}
\providecommand{\BIBforeignlanguage}[2]{{%
\expandafter\ifx\csname l@#1\endcsname\relax
\typeout{** WARNING: IEEEtran.bst: No hyphenation pattern has been}%
\typeout{** loaded for the language `#1'. Using the pattern for}%
\typeout{** the default language instead.}%
\else
\language=\csname l@#1\endcsname
\fi
#2}}
\providecommand{\BIBdecl}{\relax}
\BIBdecl

\bibitem{SCHAEFFER200727}
S.~E. Schaeffer, ``Graph clustering,'' \emph{Computer Science Review}, vol.~1,
  no.~1, pp. 27--64, 2007.

\bibitem{luczak2016asymmetry}
T.~\L{}uczak, A.~Magner, and W.~Szpankowski, ``Asymmetry and structural
  information in preferential attachment graphs,'' \emph{Random Structures and
  Algorithms}, pp. 1--23, 2019.

\bibitem{turowski2018allerton}
K.~Turowski, A.~Magner, and W.~Szpankowski, ``Compression of {D}ynamic {G}raphs
  {G}enerated by a {D}uplication {M}odel,'' in \emph{56th Annual Allerton
  Conference on Communication, Control, and Computing, Allerton 2018,
  Monticello, IL, USA, October 2-5, 2018}, 2018, pp. 1089--1096.

\bibitem{srivastava2010amphimedon}
M.~Srivastava, O.~Simakov, J.~Chapman, B.~Fahey, M.~E. Gauthier, T.~Mitros,
  G.~S. Richards, C.~Conaco, M.~Dacre, U.~Hellsten \emph{et~al.}, ``The
  amphimedon queenslandica genome and the evolution of animal complexity,''
  \emph{Nature}, vol. 466, no. 7307, p. 720, 2010.

\bibitem{sreedharan2019inferring}
J.~K. Sreedharan, A.~Magner, A.~Grama, and W.~Szpankowski, ``Inferring temporal
  information from a snapshot of a dynamic network,'' \emph{Scientific
  Reports}, vol.~9, no.~1, p. 3057, 2019.

\bibitem{kim2002asymmetry}
J.~H. Kim, B.~Sudakov, and V.~Vu, ``On the asymmetry of random regular graphs
  and random graphs,'' \emph{Random {S}tructures \& {A}lgorithms}, vol.~21, no.
  3-4, pp. 216--224, 2002.

\bibitem{psexpected}
K.~Turowski and W.~Szpankowski, ``Towards degree distribution of duplication
  graph models,'' 2019,
  \url{https://www.cs.purdue.edu/homes/spa/papers/random19.pdf}.

\bibitem{frieze2020}
A.~Frieze, K.~Turowski, and W.~Szpankowski, ``Degree distribution for
  duplication-divergence graphs: Large deviations,'' 2020, to appear in
  Proceedings of WG 2020: 46th International Workshop on Graph-Theoretic
  Concepts in Computer Science.

\bibitem{loukas2018spectrally}
A.~Loukas and P.~Vandergheynst, ``Spectrally approximating large graphs with
  smaller graphs,'' in \emph{International Conference on Machine Learning},
  Stockholm, Sweden, 2018, pp. 3243--3252.

\bibitem{liu2018global}
F.~Liu, D.~Choi, L.~Xie, and K.~Roeder, ``Global spectral clustering in dynamic
  networks,'' \emph{Proceedings of the National Academy of Sciences}, vol. 115,
  no.~5, pp. 927--932, 2018.

\bibitem{gorke2010modularity}
R.~G{\"o}rke, P.~Maillard, C.~Staudt, and D.~Wagner, ``Modularity-driven
  clustering of dynamic graphs,'' in \emph{International Symposium on
  Experimental Algorithms}.\hskip 1em plus 0.5em minus 0.4em\relax Berlin,
  Heidelberg: Springer, 2010, pp. 436--448.

\bibitem{greene2010tracking}
D.~Greene, D.~Doyle, and P.~Cunningham, ``Tracking the evolution of communities
  in dynamic social networks,'' in \emph{2010 International Conference on
  Advances in Social Networks Analysis and Mining}.\hskip 1em plus 0.5em minus
  0.4em\relax Washington, DC, USA: IEEE, 2010, pp. 176--183.

\bibitem{bair2013semi}
E.~Bair, ``Semi-supervised clustering methods,'' \emph{Wiley Interdisciplinary
  Reviews: Computational Statistics}, vol.~5, no.~5, pp. 349--361, 2013.

\bibitem{basu2002semi}
S.~Basu, A.~Banerjee, and R.~Mooney, ``Semi-supervised clustering by seeding,''
  in \emph{International Conference on Machine Learning}.\hskip 1em plus 0.5em
  minus 0.4em\relax San Francisco, CA, USA: Morgan Kaufmann Publishers Inc.,
  2002, pp. 27--34.

\bibitem{kulis2009semi}
B.~Kulis, S.~Basu, I.~Dhillon, and R.~Mooney, ``Semi-supervised graph
  clustering: a kernel approach,'' \emph{Machine Learning}, vol.~74, no.~1, pp.
  1--22, 2009.

\bibitem{li2013maximum}
S.~Li, K.~P. Choi, T.~Wu, and L.~Zhang, ``Maximum likelihood inference of the
  evolutionary history of a ppi network from the duplication history of its
  proteins,'' \emph{IEEE/ACM Transactions on Computational Biology and
  Bioinformatics (TCBB)}, vol.~10, no.~6, pp. 1412--1421, 2013.

\bibitem{navlakha2011network}
S.~Navlakha and C.~Kingsford, ``Network archaeology: uncovering ancient
  networks from present-day interactions,'' \emph{PLoS {C}omputational
  {B}iology}, vol.~7, no.~4, p. e1001119, 2011.

\bibitem{turowski2020temporal}
K.~Turowski, J.~K. Sreedharan, and W.~Szpankowski, ``Temporal ordered
  clustering in dynamic networks,'' 2020, to appear in Proceedings of IEEE
  International Symposium on Information Theory.

\bibitem{boyd2004convex}
S.~Boyd and L.~Vandenberghe, \emph{Convex Optimization}.\hskip 1em plus 0.5em
  minus 0.4em\relax Cambridge University Press, 2004.

\bibitem{pastor2003evolving}
R.~Pastor-Satorras, E.~Smith, and R.~V. Sol{\'e}, ``Evolving protein
  interaction networks through gene duplication,'' \emph{Journal of
  {T}heoretical {B}iology}, vol. 222, no.~2, pp. 199--210, 2003.

\bibitem{hofstad2016random}
R.~Van Der~Hofstad, \emph{Random graphs and complex networks}.\hskip 1em plus
  0.5em minus 0.4em\relax Cambridge: Cambridge University Press, 2016, vol.~1.

\bibitem{snapnets}
J.~Leskovec and A.~Krevl, ``{SNAP Datasets}: {Stanford} large network dataset
  collection,'' \url{http://snap.stanford.edu/data}, Jun. 2014.

\bibitem{sreedharan2020revisiting}
J.~K. Sreedharan, K.~Turowski, and W.~Szpankowski, ``Revisiting parameter
  estimation in biological networks: Influence of symmetries,'' \emph{IEEE/ACM
  Transactions on Computational Biology and Bioinformatics}, 2020,
  http://doi.org/10.1109/TCBB.2020.2980260.

\end{thebibliography}

\begin{IEEEbiography}[{\includegraphics[width=1in,height=1in,clip,keepaspectratio]{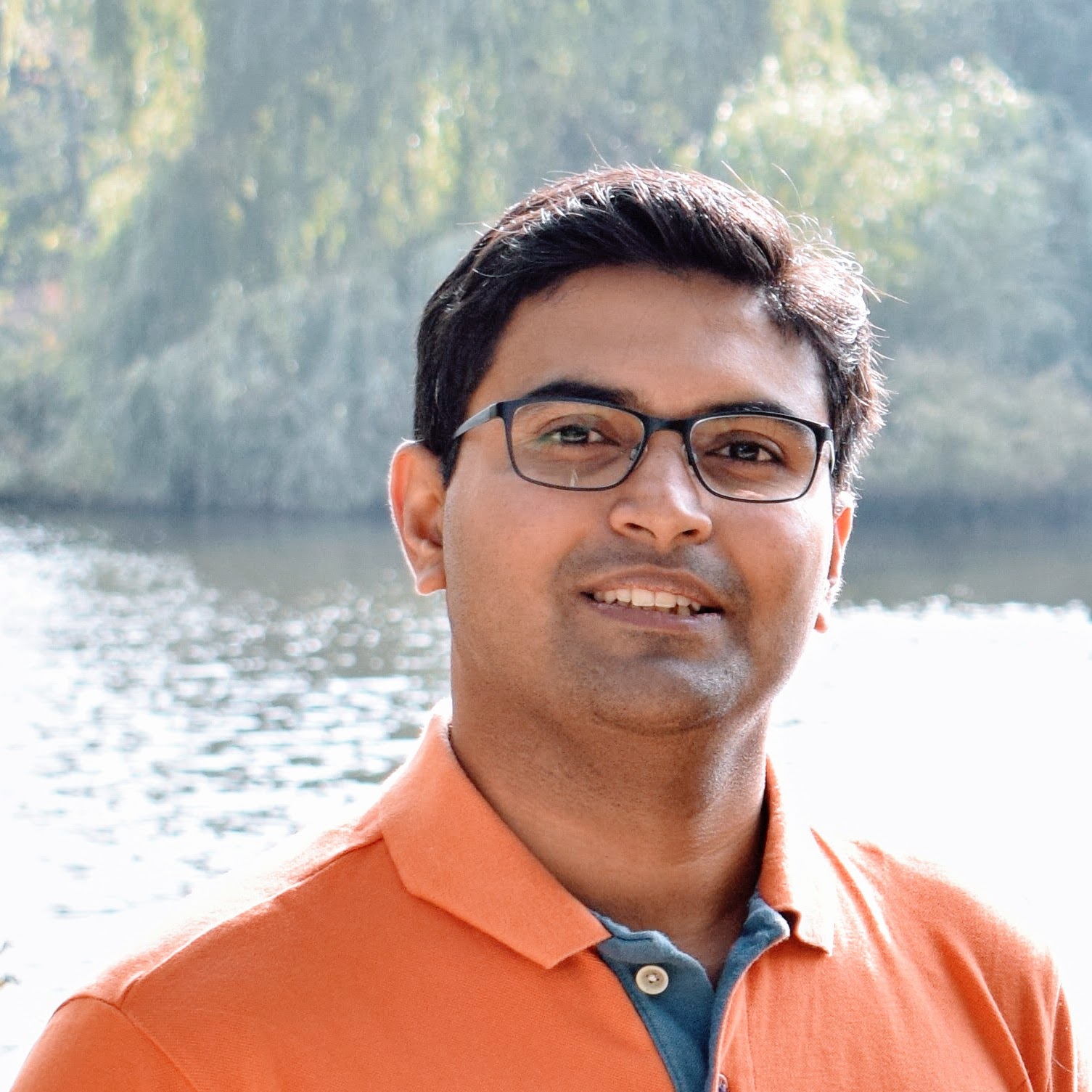}}]{Jithin K.\ Sreedharan}
is a Postdoctoral Research Associate at the NSF Center for Science of Information and Dept.\ of Computer Science in Purdue University. He received his Ph.D.\ in computer science from INRIA, France, in 2017 with a fellowship from INRIA-Bell Labs joint lab. Before that, he finished M.S.\ from Indian Institute of Science (IISc), Bangalore, in 2013, and received the best thesis award. 
His current works focus on data mining algorithms for large networks with probabilistic guarantees, statistical modeling and inference on networks, and distributed techniques for analyzing big matrices.
\end{IEEEbiography}

\begin{IEEEbiography}[{\includegraphics[width=1in,height=1in,clip,keepaspectratio]{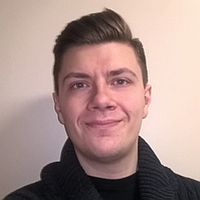}}]{Krzysztof Turowski}
is currently assistant professor at the Theoretical Computer Science Department at the Jagiellonian University, Krakow, Poland.
He received his MS and PhD degrees from Gdansk University of Technology, Poland in 2011 and 2015, respectively, both in computer science. From 2010 to 2016 he was employed at the Department of Algorithms and System Modelling at Gdansk University of Technology and from 2016 to 2018 he worked at Google as a software developer for Google Compute Engine.
From 2018 to 2019 he was a Postdoctoral Research Scholar in the NSF Center for Science of Information at Purdue University.
His research interests include graph theory (especially various models of graph coloring), analysis of algorithms and information theory.
\end{IEEEbiography}

\begin{IEEEbiography}[{\includegraphics[width=1in,height=1in,clip,keepaspectratio]{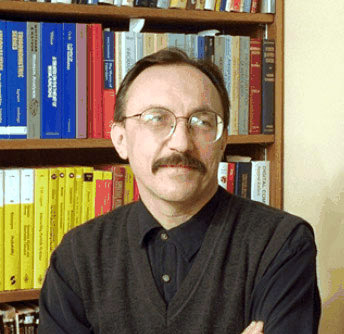}}]{Wojciech Szpankowski}
is Saul Rosen Distinguished Professor of Computer Science
at Purdue University where he teaches and conducts research in analysis of algorithms,
information theory, analytic combinatorics, data science, random structures,
and stability problems of distributed systems.
He held several Visiting Professor/Scholar positions, including
McGill University, INRIA, France,
Stanford, Hewlett-Packard Labs, Universite de Versailles, University of
Canterbury, New Zealand, Ecole Polytechnique, France, the Newton Institute,
Cambridge, UK, ETH, Zurich, and Gdansk University of Technology, Poland.
He is a Fellow of IEEE, and the Erskine Fellow.
In 2010 he received the Humboldt Research Award and in 2015
the Inaugural Arden L. Bement Jr. Award. He is also the recipient of 2020 Flajolet Lecture Prize.
He published two books: ``Average Case Analysis of
Algorithms on Sequences'', John Wiley \& Sons, 2001, and
``Analytic Pattern Matching: From DNA to Twitter'', Cambridge, 2015.
In 2008 he launched the interdisciplinary Institute for Science of
Information, and in 2010 he became the Director of the newly established NSF
Science and Technology Center for Science of Information.
\end{IEEEbiography}

% Can be used to pull up biographies so that the bottom of the last one
% is flush with the other column.
\vfill
\end{document}